\documentclass[journal]{IEEEtran}

\ifCLASSINFOpdf
\else

\fi

\usepackage{amsmath}
\usepackage{amsfonts}
\usepackage{mathrsfs}
\usepackage[dvipsnames]{xcolor}
\usepackage{graphicx}
\usepackage{float}

\usepackage{enumitem}

\usepackage{amsthm}
\theoremstyle{definition}
\newtheorem{definition}{Definition}
\newtheorem{theorem}{Theorem}
\newtheorem{corollary}{Corollary}[theorem]

\newtheorem{assumption}{Assumption}

\theoremstyle{remark}

\makeatletter
\newcommand*{\rom}[1]{\expandafter\@slowromancap\romannumeral #1@}
\makeatother

\usepackage{mathtools}

\newcommand{\cali}{\mathcal{I}}
\newcommand{\cald}{\mathcal{D}}

\newcommand{\calf}{\mathcal{F}}

\newcommand{\bbe}{\mathbb{E}}
\newcommand{\bbp}{\mathbb{P}}
\newcommand{\mh}{\mathrm{h}}
\newcommand{\mj}{\mathrm{j}}

\newcommand{\mo}{\mathrm{o}}

\renewcommand{\mp}{\mathrm{p}}

\newcommand{\mw}{\mathrm{w}}

\newcommand{\sat}{\mathrm{sat}}

\newcommand{\md}{\mathrm{d}}

\begin{document}
%
\title{Energy-efficient Connected Cruise Control with Lean Penetration of Connected Vehicles}

\author{Minghao~Shen,
Chaozhe R. He, Tamas G. Molnar, A. Harvey Bell,
and~G\'abor~Orosz 
\thanks{This research was supported by the University of Michigan’s Center for Connected and Automated Transportation through 
the US DOT grant 69A3551747105.}
\thanks{M.~Shen, A.~H.~Bell, and G.~Orosz are with the Department
of Mechanical Engineering, University of Michigan, Ann Arbor,
MI 48109, USA (e-mail: mhshen@umich.edu, ahbelliv@umich.edu, orosz@umich.edu).}
\thanks{C.~R.~He is with PlusAI Inc., Cupertino, CA 95014, USA (email: hchaozhe@umich.edu).}
\thanks{T.~G.~Molnar is with the Department of Mechanical and Civil Engineering, California Institute of Technology, Pasadena, CA 91125, USA (email: tmolnar@caltech.edu).}
\thanks{G.~Orosz is also with the Department of Civil and Environmental Engineering, University of Michigan, Ann Arbor, MI 48109, USA.}
}

\maketitle

\begin{abstract}
This paper focuses on energy-efficient longitudinal controller design for a connected automated truck that travels in mixed traffic consisting of connected and non-connected vehicles. The truck has access to information about connected vehicles beyond line of sight using vehicle-to-vehicle (V2V) communication. A novel connected cruise control design is proposed which incorporates additional delays into the control law when responding to distant connected vehicles to account for the finite propagation of traffic waves. The speeds of non-connected vehicles are modeled as stochastic processes. A fundamental theorem is proven which links the spectral properties of the motion signals to the average energy consumption. This enables us to tune controller parameters and maximize energy efficiency. Simulations with synthetic data and real traffic data are used to demonstrate the energy efficiency of the control design. It is demonstrated that even with lean penetration of connected vehicles, our controller can bring significant energy savings.
\end{abstract}


\IEEEpeerreviewmaketitle

\section{Introduction}
\IEEEPARstart{E}{nergy} saving is an everlasting theme for the truck industry, since it has the potential to provide great financial and environmental benefits to the industry and the society. Vehicle automation and connectivity technologies may bring new opportunities for energy saving. On one hand, automated vehicles can be designed such that their controllers are carefully calibrated for energy efficiency.
This includes longitudinal control systems such as adaptive cruise control~(ACC)~\cite{ard2021energy}, and predictive cruise control, where the speed profile is optimized according to the road elevation~\cite{he2016fuel,anil2020improving}.
On the other hand, vehicle-to-everything (V2X) communication facilitates information sharing and cooperation between vehicles.
This can be categorized into status-sharing, intent-sharing, agreement-seeking and prescriptive cooperation~\cite{SAEJ3216}.
A popular approach to utilize cooperation in control is cooperative adaptive cruise control (CACC)~\cite{johansson2017cooperative, Borrelli_CACC, XiaoyunLu_energy_CACC, wang2018review_CACC, wouw2019MPC} in which a platoon of connected automated vehicles is controlled to achieve great energy benefits. However, this requires all the vehicles involved to be connected and automated, which is not achievable in the near future.
In the forthcoming decades, researchers and engineers will need to deal with mixed traffic where vehicles may have various levels of connectivity and automation.

The potential energy impact of connected vehicles in mixed traffic has attracted increasing attention in the recent years~\cite{Vahidi_energy_potential_cav}, with scenarios including highways~\cite{yao2021cav_fuel_highway, ard2020cav_microsimulation, ard2021energy}, intersections~\cite{zhang2019cavimpact_intersection} and roundabouts~\cite{zhao2018cav_roundabouts}.
In this paper, we design a connected cruise controller (CCC)~\cite{gabor2016ccc} to control the longitudinal dynamics of a connected automated truck (CAT) traveling in mixed traffic which consists of connected and non-connected vehicles.
We do not require the other vehicles to be automated, and only assume lean penetration of connectivity. 
We consider the lowest level of cooperation with other connected vehicles, i.e., status-sharing cooperation, where the CAT obtains the position and speed of vehicles ahead via V2X communication.
We show that even lean penetration of connected vehicles can provide significant energy benefits for the CAT. This gives incentive to early adoptions of connectivity technologies.

In mixed traffic, with lean penetration of connectivity and low-level cooperation, a key challenge is how to acquire information about surrounding traffic with high confidence.
The CAT may be connected to vehicles in the far distance only, while surrounding non-connected vehicles may exhibit a large variety of different motions. Also, a controller that ensures high energy efficiency for one motion profile may perform poorly for another one.
One common approach is to predict the motion of preceding vehicles first, then optimize motion of the ego vehicle accordingly~\cite{Anna_prediction, Johansson_predictive_framework}. While long accurate predictions can lead to large energy savings~\cite{he2019fuel}, such predictions are hard to acquire. As uncertainties grow with the length of prediction horizon, selected optimal action could suffer large performance variations and degradation. With V2X technology, the beyond-line-of-sight information~\cite{orosz2017seeing} can potentially improve the prediction accuracy. But with lean penetration of connected vehicles, the predict-then-optimize approach may still suffer similar performance variation.

In this paper, instead of pursuing more precise prediction of transient human behavior, we minimize the energy consumption in the average sense based on vehicle trajectory data. We integrate data-driven methods and classical traffic models to optimize the energy efficiency of the connected cruise controller.
Our contributions are threefold.
First, we propose a stochastic modeling framework for the vehicles' motion and apply spectral methods to evaluate the effects of longitudinal controllers. 
The spectral properties of vehicle motion are estimated using data collected via V2X~\cite{shen2021stochastic_CCC}, and a new fundamental theorem is proven to link these properties to the average energy consumption.
Second, we account for the fact that connected vehicles may travel far ahead of the CAT, hence it may not be optimal to respond to their motions immediately.
We propose a novel CCC design which incorporates additional delays into the response to the motion of distant vehicles~\cite{shen2021delay_info} allowing the CAT to ``wait" for the velocity fluctuations to propagate along the vehicle chain.
Third, we optimize the controller gains as well as the additional delays introduced using traffic data. The optimality of the controller parameters are validated statistically using large amount of synthetic data as well as experimental data.

The remainder of this paper is organized as follows. In Section~\ref{sec:ccc_design}, we model the dynamics of the CAT, and formulate a connected cruise controller using delayed V2X information and highlight the parameters to be optimized. In Section~\ref{sec:stochastic_modeling}, we include necessary mathematical background on stochastic modeling. In Section~\ref{sec:controller_optimization}, we model the motion of preceding vehicles as stochastic processes, and establish the optimization problem which enable us to find the energy optimal controller parameters. In Section~\ref{sec:results}, we validate the optimality in terms of energy consumption based on large amount of simulations. Section~\ref{sec:concl} concludes the paper and discusses future directions.

\begin{figure}[!t]
    \centering
    \includegraphics[width=0.45\textwidth]{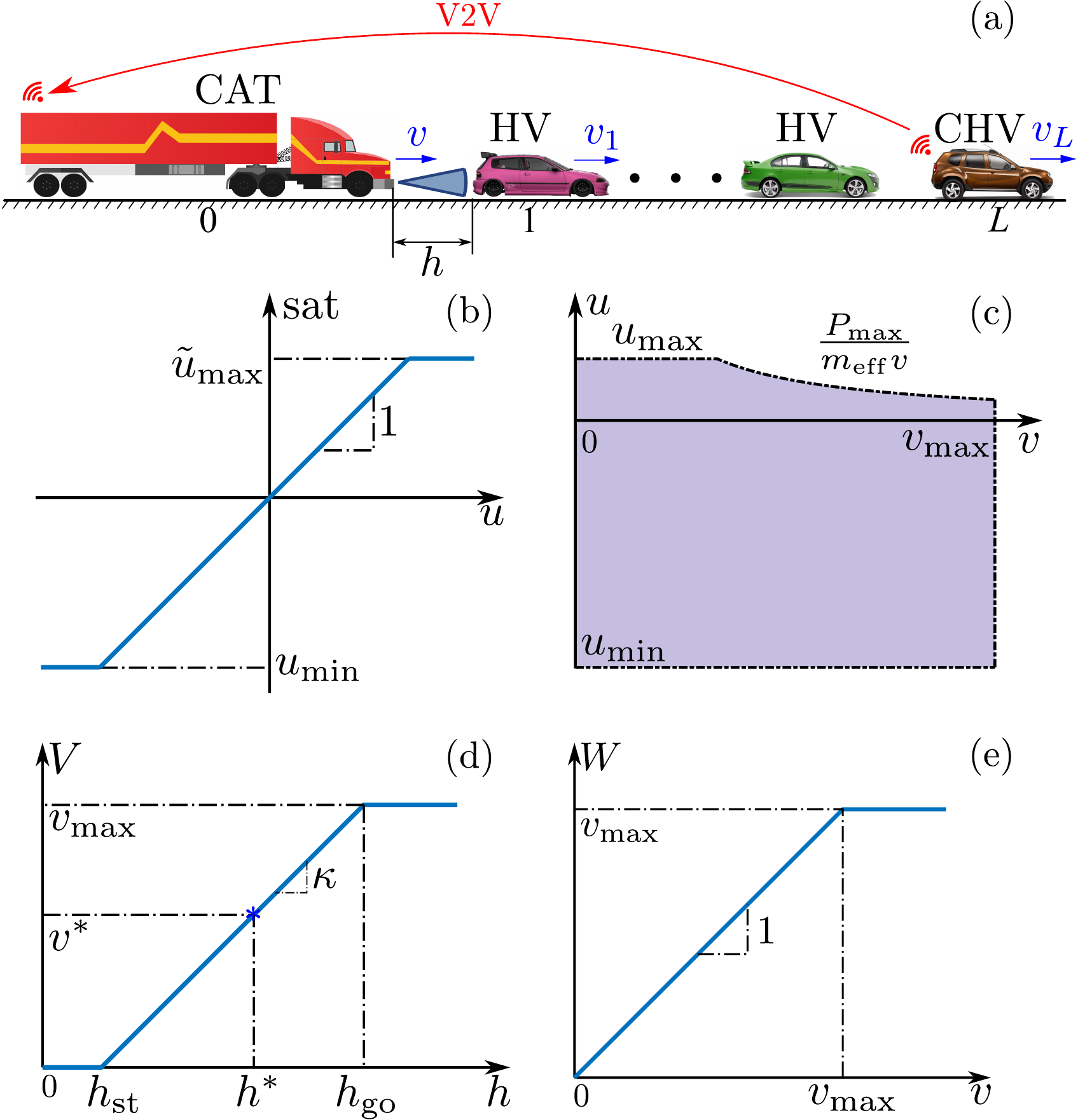}
    \caption{(a) A connected automated truck (CAT) driving in mixed traffic that consists of connected human-driven vehicles (CHVs) and non-connected human-driven vehicles (HVs). (b,c) Saturation function. (d) Range policy. (e) Speed policy.}
    \label{fig:multi_leader_diagram_with_fun}
\end{figure}

\section{Connected Cruise Control Design}\label{sec:ccc_design}

In this section, we design a longitudinal controller for a connected and automated truck (CAT) that drives in mixed traffic consisting of connected human-driven vehicles (CHVs) and non-connected human-driven vehicles (HVs); see Fig.~\ref{fig:multi_leader_diagram_with_fun}(a). The truck can measure its own speed $v$, the distance headway $h$ and the speed $v_1$ of the vehicle immediately ahead using range sensors such as camera, LiDAR or radar. Assuming there is no elevation along the road and no headwind, we can formulate the longitudinal dynamics of the truck as 
\begin{equation}\label{eqn:longitudinal_dynamics}
\begin{split}
\dot{h}(t) &= v_1(t) - v(t)\,,
\\
\dot{v}(t) &= -\frac{1}{m_{\mathrm{eff}}}\left(mg\xi + k v^2(t)\right) + \frac{T_{\mathrm{w}}(t)}{m_{\mathrm{eff}}R}\,.
\end{split}
\end{equation}
Here the dot refers to differentiation with respect to time $t$. The effective mass ${m_{\mathrm{eff}}=m+I/R^2}$ consists of the mass $m$ of the truck and the mass moment of inertia $I$ of its rotating elements. The radius of the wheels is denoted by $R$, $g$ is the gravitational constant, $\xi$ denotes the rolling resistance coefficient, and $k$ is the air drag coefficient incorporating air density and the vehicle's frontal area. In this paper, we choose $m=29484~[\mathrm{kg}]$, $I=39.9~[\mathrm{kg\ m^2}]$, $R=0.504~[\mathrm{m}]$, $\xi=0.006$, $k=3.84~[\mathrm{kg/m}]$~\cite{he2016fuel}.
We describe the nonlinear physical effects by the function
\begin{equation}
f(v) = \frac{1}{m_{\mathrm{eff}}}\left(mg\xi + k v^2\right)\,.
\end{equation}

To control the longitudinal motion of the truck, the wheel torque $T_{\mathrm{w}}$ is generated to achieve desired acceleration. When ${T_{\mathrm{w}}>0}$, the torque is provided by the powertrain, while when ${T_{\mathrm{w}}<0}$ the torque comes from the braking system. The control input $u$ is considered to be the commanded longitudinal acceleration. The effect of the control input is subject to a time delay and saturation:
\begin{equation}
\frac{T_{\mathrm{w}}(t)}{m_{\mathrm{eff}}R} = \sat\big(u(t-\sigma)\big)\,,
\end{equation}
where $\sigma$ models the delay in the powertrain system and the saturation function is given by
\begin{equation}\label{eqn:sat_definition}
\sat(u) = \left\{\begin{matrix*}[l]
u_{\min} &\mathrm{if} \quad u\leq u_{\min} \,, 
\\ 
u &\mathrm{if} \quad u_{\min}<u<\tilde{u}_{\max} \,,
\\ 
\tilde{u}_{\max} & \mathrm{if} \quad u\geq \tilde{u}_{\max} \, ,
\end{matrix*}\right.
\end{equation}
and
\begin{equation}\label{eqn:tilde_u_definition}
\tilde{u}_{\max} = \min\left\{u_{\max}, \frac{P_{\max}}{m_{\mathrm{eff}}v}\right\} \, .
\end{equation}
The saturation results from the limited available engine torque (associated with $u_{\max}$), engine power $P_{\max}$ and braking torque (associated with $u_{\min}$). They are illustrated in Fig.~\ref{fig:multi_leader_diagram_with_fun}(b,c). Here we consider the parameters
${u_{\min} = -6~[\mathrm{m}/\mathrm{s}^{2}]}$, ${u_{\max} = 2~[\mathrm{m}/\mathrm{s}^{2}]}$, ${P_{\max} = 300.65~[\mathrm{kW}]}$, ${m_{\mathrm{eff}} = 29641~[\mathrm{kg}]}$; see~\cite{he2019fuel}.

Considering the nonlinear physical effects $f(v)$, the controller consists of two terms
\begin{equation}\label{eqn:longitudinal_dynamics2}
u(t) = \tilde{f}\big(v(t)\big) + a_\md(t) \, ,
\end{equation}
where compensation term $\tilde{f}(v)$ is implemented by a lower-level controller in order to cancel the resistance term, while $a_\md$ defines the desired acceleration given by a higher-level controller that is to be designed.

With vehicle-to-everything (V2X) communication, the truck has access to information about preceding vehicles including their position and speed. In this paper, we investigate how to drive in an energy-efficient manner, and how the information from V2X communication can help to reduce the energy consumption. Define $\cali$ as the set of vehicles whose motion information the CAT has access to. This includes the vehicle immediately ahead of the truck, which is monitored by the range sensors, as well as the CHVs who share their motion information via V2X connectivity. For example, in Fig.~\ref{fig:multi_leader_diagram_with_fun}(a), the CAT has access to information about vehicles indexed $1$ and $L$, i.e., ${\cali=\{1,L\}}$. We consider longitudinal controllers of the form
\begin{equation}\label{eqn:ccc_ad}
a_\md(t) = \alpha \big(V(h(t)) - v(t)\big) + \sum_{i\in \cali} \beta_i\big(W(v_i(t-\sigma_i)) - v(t)\big) \,,
\end{equation}
where the first term is designed to maintain the headway $h$ and the second term aims to match the speed $v$ of the truck with the speed $v_i$ of the preceding vehicles. A key novelty of our controller design is to intentionally delay the response by time $\sigma_i$ in the second term. Through V2X communication, the truck can potentially connect to vehicles a few hundred meters in the distance. The behavior of the vehicles far ahead may not immediately influence the behavior of vehicles surrounding the truck, instead, the effects may be show up a few seconds later as traffic waves propagate backward with a finite speed along vehicle chains~\cite{tamas_l4dc_prediction, tamas2021TRC_Lagrangian}. Therefore, the truck can intentionally ``wait" for some time, which is determined by $\sigma_i$, and then respond to the behavior of the vehicle far in the distance.

The nonlinear functions $V$ and $W$ in \eqref{eqn:ccc_ad} are the range policy and the speed policy, respectively. The range policy
\begin{equation}\label{eqn:Vfun_def}
V(h) = \max\big\{0, \min\{\kappa(h-h_{\mathrm{st}}), v_{\max}\} \big\}\,,
\end{equation}
shown in Fig.~\ref{fig:multi_leader_diagram_with_fun}(d), is the desired speed as a function of the headway. When the headway is small, the truck intends to stop, while when the headway is large, it aims to travel at maximum speed $v_{\max}$. For headways in between, the desired speed increases linearly with gradient ${\kappa = v_{\max} / (h_{\mathrm{go}} - h_{\mathrm{st}})}$. The speed policy
\begin{equation}\label{eqn:Wfun_def}
W(v) =  \min\{v, v_{\max}\}\,,
\end{equation}
shown in Fig.~\ref{fig:multi_leader_diagram_with_fun}(e), is designed to keep the speed of truck under the speed limit $v_{\max}$ when preceding vehicles are speeding. In this paper, we set ${v_{\max} = 35~[\mathrm{m}/\mathrm{s}]}$, 
${h_{\mathrm{st}} = 5~[\mathrm{m}]}$, ${h_{\mathrm{go}} = 63.33~[\mathrm{m}]}$, yielding $\kappa = 0.6~[1/\mathrm{s}]$.

The most fundamental requirement for the controller~\eqref{eqn:ccc_ad} is to realize stable motion for the CAT.
In order to analyze the stability of the closed-loop system defined by (\ref{eqn:longitudinal_dynamics},\ref{eqn:longitudinal_dynamics2},\ref{eqn:ccc_ad}), we linearize the system around the equilibrium
\begin{equation}\label{eq:equilibrium}
h(t)  \equiv h^*\, , \quad v(t) = v_i(t) \equiv v^* = V(h^*) \, , \end{equation}
for ${i \in \cali}$. Defining the headway and speed perturbations ${\tilde{h} = h - h^*}$, ${\tilde{v} = v - v^*}$, ${\tilde{v}_i = v_i - v^*}$, 
we may obtain the linearized dynamics in the form
\begin{equation}\label{eqn:linearized_system}
\begin{split}
\dot{\tilde{h}}(t) &= \tilde{v}_1(t) - \tilde{v}(t) \, ,
\\
\dot{\tilde{v}}(t) & = \alpha \big(\kappa \tilde{h}(t-\sigma) - \tilde{v}(t-\sigma)\big)
\\
& + \sum_{i\in \cali} \beta_i\big(\tilde{v}_{i}(t-(\sigma+\sigma_i)) - \tilde{v}(t-\sigma)\big)\, .
\end{split}
\end{equation}

For analysis in frequency domain, we apply the Laplace transform with zero initial condition, which leads to
\begin{equation}\label{eqn:tf_overall_sum}
V(s) = \sum_{i\in \cali} T_{i}(s) V_{i}(s)\, .
\end{equation}
Here $V(s)$ and $V_i(s)$ denote the Laplace transforms of the speed perturbation $\tilde{v}(t)$ of the CAT and the speed perturbations $\tilde{v}_i(t)$ of the preceding vehicles, while the link transfer functions are defined as
\begin{equation}\label{eqn:tf_link}
T_1(s) = \frac{\beta_1 s + \alpha \kappa}{\cald(s)}\, , \quad
T_{i}(s) = \frac{\beta_i s e^{-s\sigma_i}}{\cald(s)}\, , 
\end{equation}
for ${i\in \cali \setminus \{1\}}$ where
\begin{equation}\label{eqn:char_eq}
\cald(s) =s^2e^{s\sigma} + \bigg(\alpha + \sum_{i\in\cali}\beta_i\bigg)s + \alpha\kappa~,
\end{equation}
gives the characteristic function.

In order to ensure that the truck is able to approach the equilibrium
\eqref{eq:equilibrium}, the linearized system \eqref{eqn:linearized_system} needs to be plant stable~\cite{zhang2016motif}. That is, all roots of the characteristic equation ${\cald(s) = 0}$ must have negative real parts.
This is satisfied when the parameters ${(\alpha, \beta_i), i\in \cali}$ are selected from the region
\begin{equation}\label{eq:plant_stable}
\begin{split}
&\alpha > 0 \, ,
\\
\underline{\omega}\sin(\underline{\omega}\sigma) - \alpha \leq &\sum_{i\in \cali} \beta_i < \overline{\omega}\sin(\overline{\omega}\sigma) - \alpha\, ,
\end{split}
\end{equation}
where $\underline{\omega}$ and $\overline{\omega}$ are the solutions of the transcendental equation ${\alpha\kappa = \omega^2\cos(\omega\sigma)}$ such that ${0<\underline{\omega}<\overline{\omega}<\frac{\pi}{2}}$. Note that the additional delay $\sigma_i$ does not influence the plant stability of the closed-loop system~\cite{shen2021delay_info}.

To evaluate the energy consumption of the CAT, we use the energy consumption per unit mass over time interval ${t\in [t_0, t_{\rm f}]}$ as metrics:
\begin{equation}\label{eqn:w_definition}
w = \int_{t_0}^{t_{\rm f}} v(t)g\big(\dot{v}(t) + f(v(t))\big) \md t \, ,
\end{equation}
where ${g(x) = \max\{x, 0\}}$, so energy is only consumed when ${u>0}$. In this paper, we consider trucks with internal combustion engines. For hybrid electric vehicles or electric vehicles, one may choose different expression for $g$~\cite{vahidi2020book}. Our goal is to find the controller parameters ${(\alpha, \beta_i, \sigma_i), i\in \cali}$ that minimize $w$ while also ensuring plant stability.

\section{Stochastic Modeling}\label{sec:stochastic_modeling}

In this section, we propose a stochastic approach where we model the motion of the preceding vehicles using stochastic processes. For simplicity, we limit our analysis to a specific family of stochastic processes, Gaussian processes, which result in physically realistic vehicle motions.

Consider a closed-loop system with dynamics (\ref{eqn:longitudinal_dynamics},\ref{eqn:longitudinal_dynamics2},\ref{eqn:ccc_ad}) where the inputs ${v_i, i\in \cali}$ are described by stochastic processes. The goal is to relate the gain parameters ${(\alpha,\beta_i), i\in\cali}$ and the delays ${\sigma_i, i\in\cali}$ through the system output $v$ to the energy consumption $w$ defined in \eqref{eqn:w_definition}. To simplify the analysis, we make three assumptions about the input processes ${v_i, i\in \cali}$: (i) they are wide-sense stationary (WSS); (ii) they are differentiable; (iii) they are Gaussian processes. We discuss these assumptions more rigorously below and relate them to spectral theory.

The stationarity assumption enables us to apply spectral analysis, and link the controller parameters to the characteristics of the output process $v$. To achieve this we need a few definitions.
\begin{definition}[Strict-sense Stationary (SSS)]
A stochastic process $\{X_t\}_{t\in T}$ is \textit{strict-sense stationary} if for any indices ${t_1,\cdots,t_k\in T}$ and sets ${A_1,\cdots,A_k}$, the probabilities
\begin{equation}
\bbp(X_{t_1 + t}\in A_1,\cdots, X_{t_k + t}\in A_k)\, ,
\end{equation}
do not depend on $t$, where $t\in T$.
\end{definition}
Specifically, choosing ${t_1=0}$ and ${k=1}$, shows that the marginal distribution of random variable $X_t$ is  time-invariant. In general, SSS is a strong requirement which is hard to satisfy. However, in many cases, the first and second moments of the distribution can provide enough information. Thus, many theories, such as spectral analysis, only require wide-sense stationarity, where stationarity is enforced only on first and second moments.
\begin{definition}[Mean and Correlations]
For a stochastic process $\{X_t\}_{t\in T}$, the \textit{mean} and the \textit{autocorrelation} are given by 
\begin{equation}\label{eq:moments}
\mu_{X}(t) = \bbe[X_t]\,, \quad R_{XX}(s, t) = \bbe[X_s X_t] \,,  
\end{equation}
where $\bbe[\cdot]$ denotes the expected value. Considering another stochastic process $\{Y_t\}_{t\in T}$ defined on the same probability space, we define the \emph{cross-correlation} as
\begin{equation}
\begin{split}
R_{XY}(s, t) = \bbe[X_s Y_t] \, .
\end{split}
\end{equation}
\end{definition}
\begin{definition}[Wide-sense Stationary (WSS)]
A stochastic process $\{X_t\}_{t\in T}$ is called wide-sense stationary if there exist a constant $m$ and a function $r(t),\, t\in T$, such that
\begin{equation}\label{eqn:wss_def}
\mu_X(t)\equiv m\,, \quad R_{XX}(s,t) = r(t-s)\,, \quad \forall s,t \in T\,.
\end{equation}
\end{definition}
That is, when $\{X_t\}_{t\in T}$ is WSS, ${R_{XX}(s, t)}$ is a function of ${(t-s)}$ and we can write ${R_{XX}(\tau) = R_{XX}(t-s)}$ without ambiguity. One may verify that autocorrelation is symmetric, that is, ${R_{XX}(s,t) = R_{XX}(t,s)}$ for a general stochastic process, yielding ${R_{XX}(\tau) = R_{XX}(-\tau)}$ for a WSS process. Similarly, the cross-correlation is also symmetric. Also note that the autocorrelation $R_{XX}(0)$ gives the second moments; cf.~\eqref{eq:moments}. We assume that speed perturbations of the preceding vehicles $\tilde{v}_i$ are WSS, that is, ${v_i = v^* + \tilde{v}_i}$ where $v^*$ denotes the equilibrium speed and ${\mu_{\tilde{v}_i}=0}$, for all ${i\in \cali}$. 

For a signal that satisfies WSS condition, we can apply spectral analysis and determine the input/output relationship for linear time-invariant (LTI) systems. Such analysis utilizes the power spectral density which can be defined via the continuous-time Fourier transform of the WSS process.
\begin{definition}[Power spectral density~\cite{hajek2015random}]
For WSS process $X_t$, the \textit{power spectral density} is the Fourier transform of autocorrelation function:
\begin{equation}
S_{XX}(\omega) = \calf[R_{XX}(\tau)] = \int_{-\infty}^{\infty} R_{XX}(\tau)e^{-\mj\omega\tau}\md \tau\, ,
\end{equation}
where $\omega$ denotes the angular frequency.
\end{definition}
Since ${R_{XX}(\tau) = R_{XX}(-\tau)}$, the power spectral density $S_{XX}(\omega)$ is a non-negative real number and one can also show that ${S_{XX}(\omega) = S_{XX}(-\omega)}$. 
For LTI systems with input being a WSS process, the power spectral density of the output process can be calculated using the following theorem~\cite{hajek2015random}.
\begin{theorem}[Spectral Analysis of LTI Systems]\label{thm:psd_LTI}
\textit{
For a linear time invariant system with transfer function $G(s)$, if the input signal $X_t$ is a WSS process, then the output signal $Y_t$ is also WSS. The first and second moments of $Y_t$ are given by
\begin{equation}
\mu_{Y} = G(0)\mu_{X}\, , \quad
S_{YY}(\omega) = |G(\mj \omega)|^2 S_{XX}(\omega)\, .
\end{equation}
}
\end{theorem}
The proof can be found in Chapter~8.2 of~\cite{hajek2015random}.

Similarly, the \textit{cross power spectral density} can be defined as the Fourier transform of the cross-correlation function
\begin{equation}
S_{XY}(\omega) = \calf[R_{XY}(\tau)]\, ,
\end{equation}
which may be a complex number. The following theorem defines the input/output relationship of signals passing through different LTI systems~\cite{hajek2015random}.
\begin{theorem}\label{thm:csd_LTI}
\textit{
Given two signals $X_t$ and $Y_t$ separately passing through two LTI systems with transfer functions $G_1(s)$ and $G_2(s)$, respectively,  the cross power spectral density of the corresponding outputs  $Z_t$ and $P_t$ is
\begin{equation}\label{eqn:cpsd_two_LTI}
S_{Z P}(\omega) = G_1(\mj \omega)G_2^*(\mj \omega)S_{XY}(\omega)\, ,
\end{equation}
where star denotes complex conjugate.
}
\end{theorem}
The proof can be found in Chapter 8.2 of~\cite{hajek2015random}.

In practice, it is reasonable to assume that, the speeds of preceding vehicles are continuous and differentiable, i.e., the acceleration is continuous.
Specifically, for a WSS process $X_t$, we have the following properties of the time derivative $\dot{X}_t$:
\begin{enumerate}[label=(\alph*)]
\item $X_t$ and $\dot{X}_t$ are jointly WSS
\item $\mu_{\dot{X}}(t) = 0$
\item $R_{\dot{X}X}(\tau) = \frac{\md}{\md \tau}R_{XX}(\tau)= -R_{X\dot{X}}(\tau)$
\item $R_{\dot{X}X}(0) = R_{X\dot{X}}(0)=0$\
\item $R_{\dot{X}\dot{X}}(\tau)=-\frac{\md ^2}{\md \tau^2}R_{XX}(\tau)$
\end{enumerate}
The proof can be found in Chapter 7.2 of~\cite{hajek2015random}.

Apart from being differentiable, the speeds of preceding vehicles are assumed to be Gaussian processes. This simplifies the analysis and enables us to derive analytical results.
\begin{definition}[Gaussian Process (GP)]
A stochastic process $\left\{X_t\right\}_{t\in T}$ is a \textit{Gaussian process} if for every finite set of indices ${t_1,\cdots, t_k \in T}$, ${X(t_1,\cdots, t_k) = \left(X_{t_1}, \cdots, X_{t_k}\right)}$ is multivariate Gaussian random variable.
\end{definition}
Gaussian process has the following nice properties~\cite{hajek2015random}.
\begin{enumerate}[label=(\alph*)]
\item Gaussian process is uniquely determined by its mean function and autocorrelation function.
\item If a Gaussian process is WSS, then it is SSS.
\item For a linear system, if the input signal is a Gaussian process, then the output is also a Gaussian process.
\item If a Gaussian process $X_t$ is mean square differentiable, then $\dot{X}_t$ is also a Gaussian process.
\end{enumerate}

In this section, we have established necessary properties of the motions of preceding vehicles by assuming them to be stationary differentiable Gaussian processes. We are now ready to apply spectral analysis on the closed-loop linearized system \eqref{eqn:linearized_system} to derive the distribution of the CAT's motion, as well as the average energy consumption.
This allows us to optimize the controller parameters ${(\alpha, \beta_i, \sigma_i), i\in \cali}$ for energy efficiency based on the data obtained from V2X connectivity.

\section{Data-driven Controller Optimization}\label{sec:controller_optimization}

In this section, we propose a method to determine the energy-optimal parameters for the proposed controller using traffic data. First we derive an optimization problem assuming oracle knowledge about the spectral density of the preceding vehicles' speed. Then we introduce two estimators for the cross power spectral density, and finally formalize the data-driven controller optimization method.

\subsection{Optimization with Oracle Knowledge}

Here we utilize the theory introduced in the previous section, to apply spectral analysis for the linearized system \eqref{eqn:linearized_system}, derive analytical expression for the expectation of the energy consumption defined in \eqref{eqn:w_definition}, and formulate an optimization problem to determine energy-optimal controller parameters.
We achieve these results under the following assumption.

\begin{assumption}\label{assum:Gaussian}
The inputs ${\tilde{v}_i, i\in \cali}$ are WSS, mean-square differentiable Gaussian processes with zero mean.
\end{assumption}

For linearized system~\eqref{eqn:linearized_system}, the nonlinear physical effects $f(v)$ and saturation $\sat(\cdot)$ are dismissed. Therefore we consider the surrogate energy consumption model
\begin{equation}\label{eqn:linear_energy_model}
\bar{w} = \int_{t_0}^{t_{\mathrm{f}}} v(t) g\big(\dot{v}(t)\big)\mathrm{d}t\, ,
\end{equation}
cf.~\eqref{eqn:w_definition}. Then, the following theorem provides analytical expression for the energy consumption.

\begin{theorem}\label{thm:energy}
\textit{
Consider the linearized dynamics~\eqref{eqn:linearized_system} around equilibrium~\eqref{eq:equilibrium} under Assumption~\ref{assum:Gaussian}.
The expectation of the energy consumption defined in \eqref{eqn:linear_energy_model} is given by
\begin{equation}
\bbe[\bar{w}] = (t_{\rm f} - t_0)\frac{v^* }{\sqrt{2\pi}}\vartheta\,,
\end{equation}
where
\begin{equation}\label{eq:theta2}
\vartheta^2 = \frac{1}{\pi} \sum_{i,j\in \cali}\int_{0}^{\infty}\omega^2 T_i(\mj \omega)T_j^*(\mj \omega)S_{\tilde{v}_i\tilde{v}_j}(\omega)\md \omega\,.
\end{equation}
represents the variance of $\dot{v}$, which is equal to the variance of $\dot{\tilde{v}}$.
}
\end{theorem}

\begin{proof}
Under Assumption~\ref{assum:Gaussian}, the output $\tilde{v}$ is a WSS Gaussian process with zero mean. According to \eqref{eqn:tf_overall_sum}, the output signal $\tilde{v}$ can be decomposed into response $\eta_i$ to each input signal $\tilde{v}_i$:
\begin{equation}
\tilde{v}(t) = \sum_{i\in \cali} \eta_i(t)\, . 
\end{equation}

In time domain, we have
\begin{equation}\label{eqn:rvv}
\begin{split}
R_{\tilde{v}\tilde{v}}(\tau) &= \bbe[\tilde{v}(t)\tilde{v}(t+\tau)] \\
& = \sum_{i,j\in \cali} \bbe[\eta_i (t) \eta_j(t+\tau)] \\
& = \sum_{i,j\in \cali} R_{\eta_i \eta_j}(\tau) \, . \\
\end{split}
\end{equation}
Taking the Fourier transform and noting that ${\mathcal{F}[R_{\tilde{v}\tilde{v}}(\tau)]=S_{\tilde{v}\tilde{v}}(\omega)}$ and ${\mathcal{F}[R_{\tilde{\eta}_i\tilde{\eta}_j}]=S_{\tilde{\eta}_i \tilde{\eta}_j}(\omega)}$, we obtain
\begin{equation}
\begin{aligned}
S_{\tilde{v}\tilde{v}}(\omega) & = \sum_{i,j\in \cali} S_{\eta_i \eta_j}(\omega) \\
& = \sum_{i,j\in \cali} T_i(\mj\omega) T^*_j(\mj\omega)S_{\tilde{v}_i\tilde{v}_j}(\omega) \, ,
\end{aligned}
\end{equation}
where in the last step we used Theorem~\ref{thm:csd_LTI}. Note that when ${i=j}$, we have ${S_{\eta_i \eta_j} = |T_{ij}(\mj \omega)|^2 S_{\tilde{v}_i\tilde{v}_j}(\omega)}$.

The speed perturbation of the truck $\tilde{v}$ as well as its derivative $\dot{\tilde{v}}$ are WSS Gaussian processes with zero mean.
Let us consider the second moments
\begin{equation}
\varsigma^2 = R_{\tilde{v}\tilde{v}}(0)\, ,\quad \vartheta^2 = R_{\dot{\tilde{v}}\dot{\tilde{v}}}(0)\, .
\end{equation}
Since for WSS process ${R_{\dot{X}\dot{X}}(\tau) = -\frac{\md^2}{\md \tau^2}R_{XX}(\tau)}$, we can express the variance of $\dot{\tilde{v}}$ as
\begin{equation}
\begin{split}
\vartheta^2 & = R_{\dot{\tilde{v}}\dot{\tilde{v}}}(0)\\
& =\left.-\frac{\md^2}{\md \tau^2} R_{\tilde{v}\tilde{v}}(\tau) \right|_{\tau = 0}\\
& = \calf^{-1}[\omega^2 S_{\tilde{v}\tilde{v}}(\omega)]|_{\tau=0}\\
& = \left.\frac{1}{2\pi}\int_{-\infty}^{\infty} \omega^2 S_{\tilde{v}\tilde{v}}(\omega) e^{\mj\omega\tau}\md \omega\right|_{\tau=0} \, .\\
\end{split}
\end{equation}
Since $S_{\tilde{v}\tilde{v}}(\omega) = S_{\tilde{v}\tilde{v}}(-\omega)$, we have
\begin{equation}
\begin{aligned}
\vartheta^2 &= \frac{1}{\pi} \int_{0}^{\infty}\omega^2 S_{\tilde{v}\tilde{v}}(\omega)\md\omega\\
&=\frac{1}{\pi} \sum_{i,j\in \cali}\int_{0}^{\infty}\omega^2 T_i(\mj \omega)T_j^*(\mj \omega)S_{\tilde{v}_i\tilde{v}_j}(\omega)\md \omega.
\end{aligned}
\end{equation}
Thus, considering $R_{\tilde{v}\dot{\tilde{v}}}(0)=0$, we can write down the joint distribution of $v = \tilde{v} + v^*$ and $\dot{v} = \dot{\tilde{v}}$ as follows
\begin{equation}\label{eqn:p_v_vdot}
p(v, \dot{{v}}) = \frac{1}{2\pi \varsigma \vartheta}\exp\left(-\frac{(v-v^*)^2}{2\varsigma^2}-\frac{\dot{v}^2}{2\vartheta^2}\right)\, .
\end{equation}

From SSS assumption, the distributions of $v(t)$ and $\dot{v}(t)$ are time-invariant. The mean value of the energy consumption $\bar{w}$ defined in~\eqref{eqn:linear_energy_model} can be calculated as
\begin{align}
\bbe[\bar{w}] &= \int_{t_0}^{t_{\rm f}}\md t \int_{-\infty}^{\infty}\md v \int_{-\infty}^{\infty} v g(\dot{v})p(v, \dot{v})\md \dot{v} \nonumber
\\
& = (t_{\rm f} - t_0) \int_{-\infty}^{\infty}\md v \int_{0}^{\infty} v \dot{v} p(v, \dot{v})\md \dot{v} 
\\
&= (t_{\rm f} - t_0) \frac{1}{2\pi \varsigma \vartheta}\left[\int_{-\infty}^{\infty}v\exp\left(-\frac{(v-v^*)^2}{2\varsigma^2}\right)\md v\right] \nonumber
\\
&\qquad\qquad\qquad\, \times \left[ \int_{0}^{\infty} \dot{v} \exp\left(-\frac{\dot{v}^2}{2\vartheta^2}\right)\md \dot{v}\right] \nonumber
\\
& = (t_{\rm f} - t_0)\frac{v^* }{\sqrt{2\pi}}\vartheta \, . \nonumber
\end{align}
This completes the proof.
\end{proof}

As consequence of Theorem~\ref{thm:energy}, parameters that minimize $\vartheta^2$, the variance of $\dot{v}$, also minimize the average energy consumption. Note that in the optimal control literature~\cite{kirk2004optimal}, people often use the square magnitude of acceleration as cost function, and the corresponding control problem is referred to as ``energy optimal control". Here, for the first time, we proved that this quantity is proportional to the square root of average energy consumption, when the energy is defined in \eqref{eqn:linear_energy_model}. 

In this paper, we fix ${\alpha=0.4~[1/\mathrm{s}]}$ for safety considerations~\cite{ge2018experimental}, and search for the optimal parameters ${(\beta_i, \sigma_i), {i \in \cali}}$.
We find their values by solving an optimization problem summarized in the following corollary of Theorem~\ref{thm:energy}. Note that the results can be easily extended to include $\alpha$.
\begin{corollary}
\textit{
Consider the linearized dynamics~\eqref{eqn:linearized_system} around equilibrium~\eqref{eq:equilibrium} under Assumption~\ref{assum:Gaussian}.
The optimal values of the control parameters $(\beta_i, \sigma_i), {i \in \cali}$ that minimize the expectation of the energy consumption \eqref{eqn:linear_energy_model} while maintaining plant stability can be found by solving the optimization problem
\begin{align}\label{eqn:optimization_continuous}
\underset{(\beta_i, \sigma_i)}{\min} J= & \sum_{i,j\in \cali}\int_{0}^{\infty}\omega^2 T_i(\mj \omega) T_j^*(\mj \omega) S_{\tilde{v}_i \tilde{v}_j}(\omega)\md \omega \nonumber~, \\
\mathrm{s.t.}\quad & (\beta_i) \in \Omega~,
\end{align}
where $\Omega$ is the plant stable region defined in \eqref{eq:plant_stable}.
}
\end{corollary}
Power spectral densities of speed perturbations of preceding vehicles are included in the objective function of \eqref{eqn:optimization_continuous} when ${i=j}$. Furthermore, the cross power spectral densities in the objective function capture the correlations between the speed perturbation signals of the preceding vehicles when ${i\neq j}$. Such correlations are especially significant in dense traffic where the motions of subsequent vehicles are typically strongly coupled.
For example, Newell's car-following model considers the speed of a following vehicle as delayed copy of the vehicles ahead~\cite{newell2002car-following}. Computing the cross power spectral density allows us to capture and utilize such strong correlations to improve energy efficiency.

\subsection{Data-driven Optimization}

In practice, the true value of the cross power spectral density $S_{\tilde{v}_i\tilde{v}_j}(\omega)$ is unknown. Instead, we need to estimate it from finite amount of data sampled in discrete time. In this paper, we utilize two estimators: periodogram and Welch's method~\cite{shumway2017time}. 

Consider observation data of two velocity signals $\{v_i(t_k)\}_{k=0}^{N-1}$ and $\{v_j(t_k)\}_{k=0}^{N-1}$ for time instances ${t_k = k\Delta t}$. First, we subtract from each term their sample mean to get the centralized data $\{\tilde{v}_i(t_k)\}_{k=0}^{N-1}$ and $\{\tilde{v}_j(t_k)\}_{k=0}^{N-1}$.
Let $\{\tilde{V}_i(\omega_k)\}_{k=0}^{N-1}$ and $\{\tilde{V}_j(\omega_k)\}_{k=0}^{N-1}$, $\omega_k =\frac{2\pi k}{N\Delta t}$ be the discrete-time Fourier transforms of $\{\tilde{v}_i(t_k)\}_{k=0}^{N-1}$ and $\{\tilde{v}_j(t_k)\}_{k=0}^{N-1}$, respectively. The periodogram method estimates the cross-spectral density with
\begin{equation}\label{eqn:periodogram:cpsd}
\hat{S}_{\tilde{v}_i \tilde{v}_j} (\omega_k) = \frac{2\Delta t}{N} \tilde{V}_i(\omega_k) \tilde{V}_j^*(\omega_k)\, .
\end{equation}
When ${i=j}$, it reduces to the one-sided estimator of power spectral density
\begin{equation}\label{eqn:periodogram_psd}
\hat{S}_{\tilde{v}_i \tilde{v}_i}\left(\omega_k\right) = \frac{2\Delta t}{N}\big|\tilde{V}_i (\omega_k)\big|^2 \, .
\end{equation}

The periodogram estimator is asymptotically unbiased, but suffers from high variance. We can apply the following methods to reduce the variance. In the time domain, we can split the original signals into segments, calculate the periodogram of each segment, then average them for each frequency. In the frequency domain, we can apply window functions, such as Hamming window, and calculate the periodogram for windowed signals. Welch's method~\cite{welch_method} combines the two solutions together. First we split each original signal into overlapping segments, with 50\% overlap ratio. Then we apply window function to each segment and calculate the periodogram for each windowed segment. Finally, we average over all the periodograms for each frequency. Welch's method can achieve lower variance, but at the cost of frequency resolution. We will see later in Section \ref{sec:choice_estimator} that the variance reduction property can help us to obtain better controller gains $\beta_i$, but the low resolution poses limitations on getting appropriate information delay $\sigma_i$.

It is straightforward to replace the power spectral density $S_{\tilde{v}_i \tilde{v}_j}$ in \eqref{eqn:optimization_continuous} with periodogram estimator, and rewrite the optimization problem for discrete-time observations:
\begin{equation}\label{eqn:optimization_periodogram}
\begin{aligned}
\underset{(\beta_i, \sigma_i)}{\min} & J \approx \frac{2\Delta t}{N}  \sum_{k=0}^{N-1} \sum_{i,j\in \cali}\omega_k^2 T_i(\mj \omega_k) T_j^*(\mj\omega_k)\tilde{V}_i(\omega_k)\tilde{V}_j^*(\omega_k)\\
& \;\; = \frac{2\Delta t}{N} \sum_{k=0}^{N-1}\omega_k^2\left|\sum_{i\in \cali} T_i(\mj \omega_k)\tilde{V}_i( \omega_k) \right|^2 \\
& \;\; = \frac{2\Delta t}{N} \sum_{k=0}^{N-1}\omega_k^2 \left|\tilde{V}(\omega_k)\right|^2 \\
\mathrm{s.t.}\; & (\beta_i) \in \Omega.
\end{aligned}
\end{equation}

The optimization problem is similar for Welch's method. We only need to substitute the cross power spectral density $S_{\tilde{v}_i \tilde{v}_j}(\omega)$ in \eqref{eqn:optimization_continuous} with estimation results from Welch's method. Notice that we do not need to know the equilibrium velocity $v^*$ in our data-driven method.

We remark that by choosing the periodogram as the spectral estimator, we recover the optimization framework in our previous work~\cite{he2019fuel}, where the energy optimal controller parameters are selected by minimizing ${\sum_{k=1}^{N-1} \omega_k^2 \big|\tilde{V}(\omega_k)\big|^2}$, which is identical to our optimization problem \eqref{eqn:optimization_periodogram} if we choose the periodogram as spectral estimator. 
The method presented in this paper gives solid theoretical justifications and extend the framework to allow further improvement by choosing a better spectral estimator, e.g., Welch's method.

\section{Numerical Results}\label{sec:results}

In this section, we evaluate the optimization method proposed in the previous section using both synthetic data and real traffic data. With synthetic data, we have access to the underlying ground truth distribution of trajectories, which enables us to verify our theory. With real traffic data, we can show the potential of applying the proposed method in the real world.
The evaluation scenarios include adaptive cruise control and connected cruise control with CHVs ahead of the CAT in the traffic.
To showcase the potential in energy saving of the proposed design for lean penetration of connectivity, we consider the scenario where the truck is connected to a CHV $L$ vehicles ahead.

\subsection{Simulation Dataset}\label{sec:dataset}

We consider two kinds of traffic data: driving in free-flow conditions and in traffic congestion~\cite{ge2018experimental}. These are shown in Fig.~\ref{fig:traj_real_synthetic}(a) and (b) respectively. In the second case, the leading vehicle frequently makes mild brakes, and the following vehicles have increasingly harsh brakes because of the string instability of human drivers~\cite{Feng2019, Ploeg2014}. This string instability implies that, speed fluctuations of vehicles in the distance may be milder than those of vehicles closer to the truck. With V2X connectivity, the CAT may not only respond the immediate preceding vehicle, but also the vehicles far in the distance. On the other hand, the observed phase-lags between the braking events motivate us to introduce the additional delays $\sigma_i$ in \eqref{eqn:ccc_ad}. That is, it takes time for the behavior of preceding vehicles to affect the CAT, therefore, it shall ``wait" before responding to vehicles far in the distance.

We generate synthetic speed trajectories for the preceding vehicles according to the stochastic modeling assumptions in Section~\ref{sec:stochastic_modeling}. In particular, we first generate the speed profile of vehicle $L$, and then simulate the following vehicles using a car-following model. According to Assumption~\ref{assum:Gaussian} in Section \ref{sec:stochastic_modeling}, the stochastic process ${\tilde{v}_L(t) = v_L(t) - v^*}$ is a mean zero differentiable Gaussian process. In this paper, we choose Mat\'ern kernel~\cite{Rasmussen2006gpml}
\begin{equation}\label{eq:matern_kernel}
R_{\tilde{v}_L\tilde{v}_L}(\tau) = C^2 \frac{2^{1-\nu}}{\Gamma(\nu)}\left(\sqrt{2\nu}\frac{\tau}{\rho}\right)^{\nu} K_{\nu}\left(\sqrt{2\nu} \frac{\tau}{\rho}\right)\,,
\end{equation}
for this Gaussian process, where $\Gamma$ is the Gamma function, $K_\nu$ is the modified Bessel function of the second kind, $C$, $\rho$ and $\nu$ are positive parameters. In our simulations, we choose ${v^* = 25~[{\rm m}/{\rm s}]}$, ${C=1}$, ${\rho = 5}$, and ${\nu=\frac{5}{2}}$. It can be proven that this process is second-order mean square differentiable~\cite{Rasmussen2006gpml}.

To generate speed profiles for the following vehicles we use the optimal velocity model (OVM) for human driver behavior:
\begin{equation}\label{eqn:human_driver_model}
\begin{split}
\dot{h}_i(t) &= v_{i+1}(t) - v_i(t) \\
\dot{v}_i(t) &= \alpha_\mh(V(h_i(t-\sigma_\mh)) - v_i(t-\sigma_\mh)) \\
& + \beta_\mh(W(v_{i+1}(t - \sigma_\mh))-v_i(t-\sigma_\mh)),
\end{split}
\end{equation}
for all ${i = 1,\ldots, L-1}$, where the range policy is given by
\begin{equation}\label{eqn:Vhfun_def}
V_{\mh}(h) =\left\{\begin{matrix*}[l]
0 & \mathrm{if} \quad h\leq h_{\mathrm{st}} \, ,
\\
\kappa_{\mh}(h-h_{\mathrm{st}}) & \mathrm{if} \quad h_{\mathrm{st}}< h < h_{\mathrm{go}} \, ,
\\
v_{\max} & \mathrm{if} \quad h\geq h_{\mathrm{go}} \, ,
\end{matrix*}\right.
\end{equation}
and $W$ is defined in \eqref{eqn:Wfun_def}. We choose parameters ${\alpha_\mh = 0.2~[1/{\rm s}]}$, ${\beta_\mh=0.8~[1/{\rm s}]}$, ${\kappa_\mh=1.0~[1/{\rm s}]}$, ${\sigma_\mh = 1.0~[{\rm s}]}$, and ${L=8}$. Fig.~\ref{fig:traj_real_synthetic}(c) shows the corresponding synthetic speed trajectories.

\begin{figure}[t]
    \centering
    \includegraphics[width=0.48\textwidth]{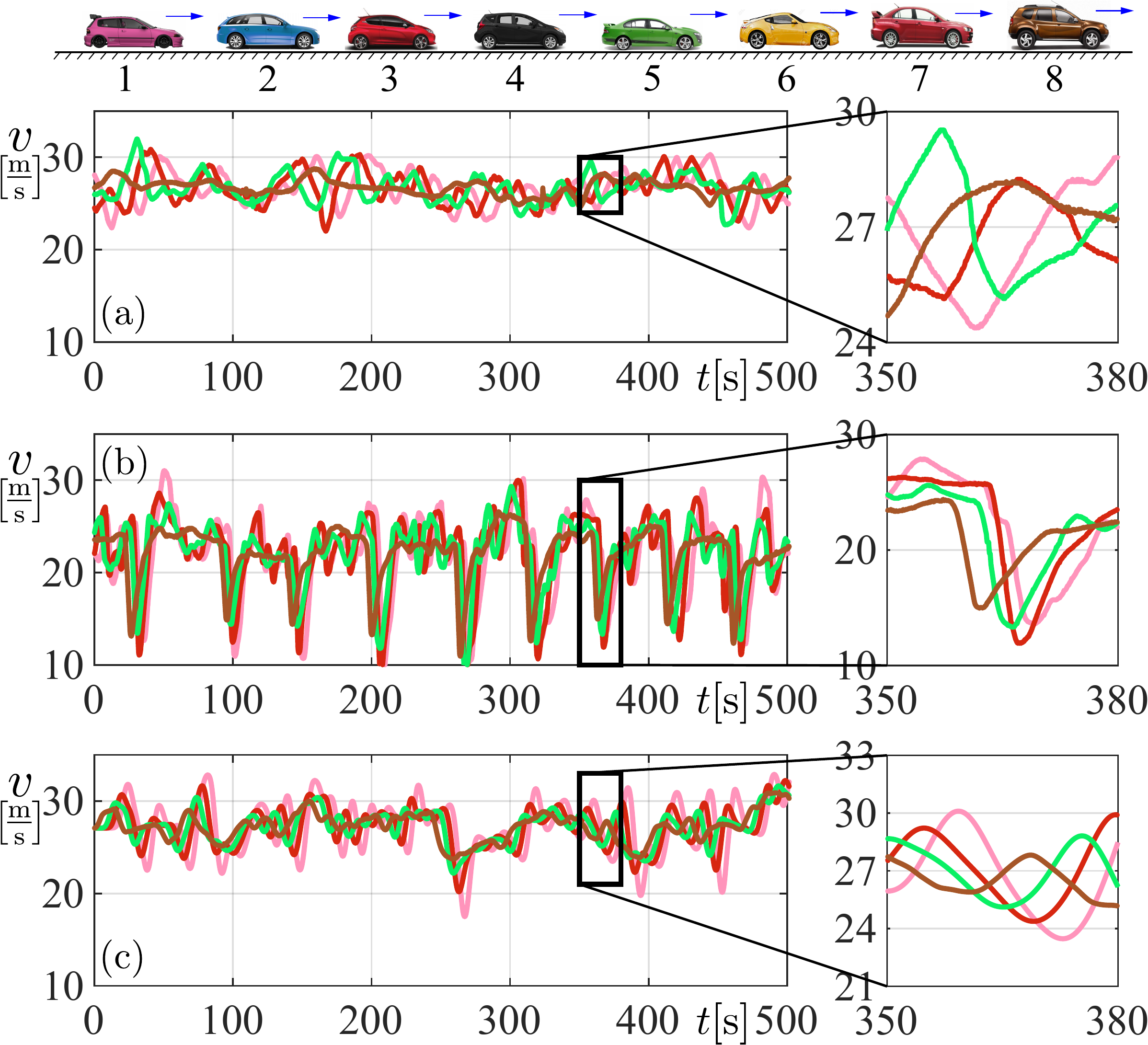}
    \caption{Velocity profiles for vehicles $8$ (brown), $5$ (green), $3$ (red), $1$ (pink). (a) Experimental velocity data in free-flow traffic condition. (b) Experimental velocity data in congestion. (c) Synthetic data.}
    \label{fig:traj_real_synthetic}
\end{figure}

\subsection{Benefits of Connectivity}\label{sec:benefit_connectivity}

In this section, we compare the energy consumption for scenarios with and without connectivity in the traffic. Based on the synthetic dataset introduced in the previous section, we apply a \emph{cross evaluation} method to evaluate the proposed controllers. This consists of two steps: observation and testing. In observation step, we observe the speed profile of the preceding vehicle, estimate the spectral density, and solve optimization problem \eqref{eqn:optimization_continuous} to get the optimal controller parameters. In testing step, we simulate the truck for different preceding vehicle speed profiles, using the optimal parameters calculated in the observation step. The speed profile of the preceding vehicle in testing step shall have the same distribution as that in observation step. In particular, for synthetic dataset, we create 101 candidate speed profiles for the preceding vehicle according to Section \ref{sec:dataset} and arbitrarily pick one for observation and another one for testing. Therefore there are $101\times 100$ observation-evaluation pairs.

Since the theoretical results are derived based on linearization, we conduct simulations for both linear and nonlinear systems. In the former case, we simulate the linearized system~\eqref{eqn:linearized_system} where nonlinear physical effects $f(v)$ and $\sat(\cdot)$ are dropped. Consequently the energy consumption model~\eqref{eqn:linear_energy_model} is applied for evaluation. On the other hand, when simulating the nonlinear system~\eqref{eqn:longitudinal_dynamics}\eqref{eqn:longitudinal_dynamics2}\eqref{eqn:ccc_ad}, we take into account all the nonlinear effects, and use \eqref{eqn:w_definition} to calculate the energy consumption.

When there is no connected vehicle in the traffic, the truck can only collect information about the motion of the vehicle immediately ahead. We refer to this as adaptive cruise control (ACC). The corresponding energy consumption will serve as a benchmark, and will be compared to the CCC controllers which exploit V2X connectivity. The acceleration command of adaptive cruise control is given by
\begin{equation}\label{eqn:acc_ad}
a_\md(t) = \alpha \big(V(h(t)) - v(t)\big) + \beta_1\big(W(v_1(t) - v(t)\big) \,,
\end{equation}
cf.~\eqref{eqn:ccc_ad}. Thus, \eqref{eqn:tf_overall_sum}\eqref{eqn:tf_link}\eqref{eqn:char_eq} yields
\begin{equation}
V(s) = T(s)V_1(s),\quad T(s) = \frac{\beta_1 s + \alpha \kappa}{s^2e^{s\sigma} + (\alpha+\beta_1)s + \alpha \kappa}\;.
\end{equation}

\begin{figure}[h]
    \centering
    \includegraphics[width=0.44\textwidth]{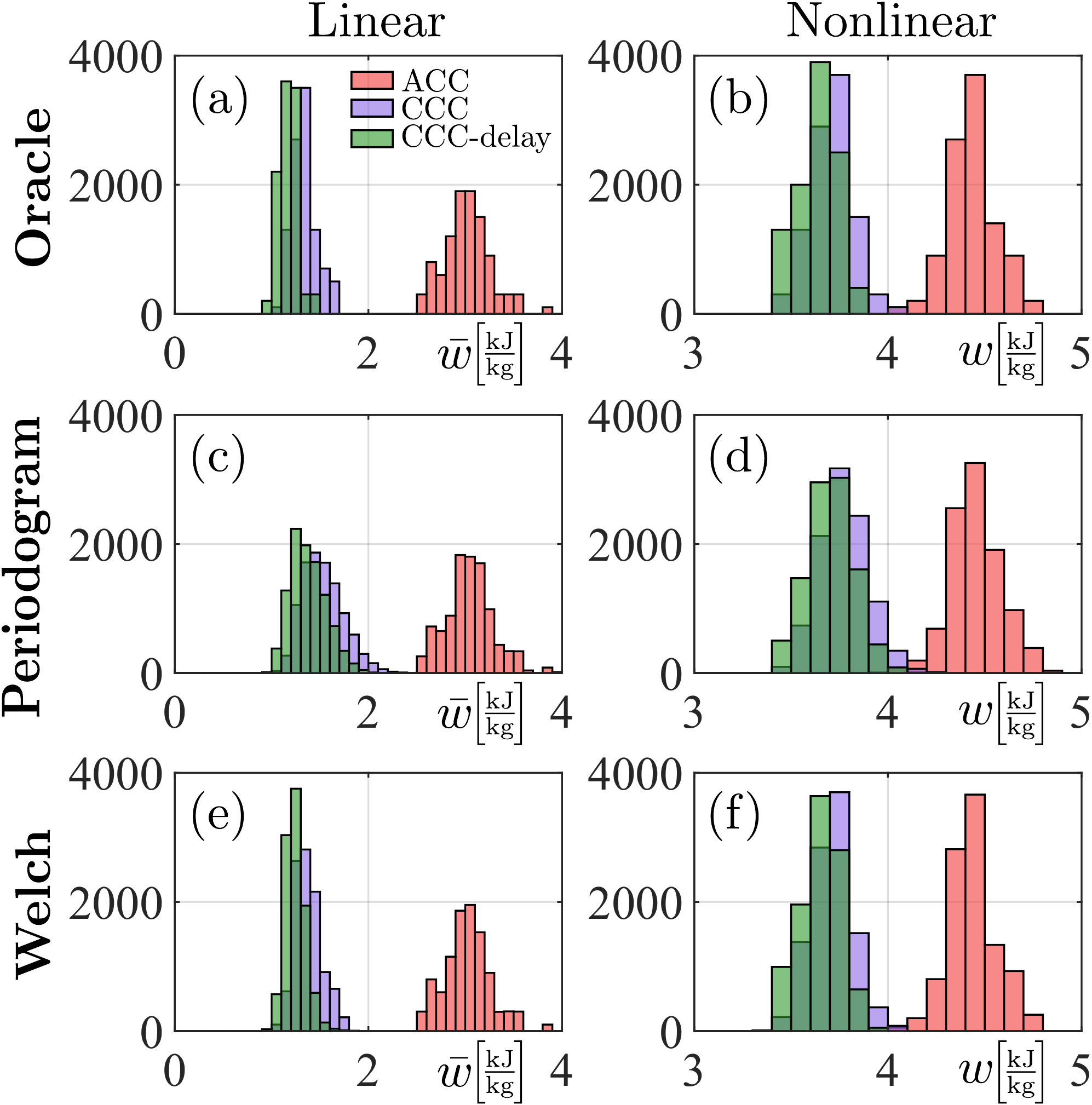}
    \caption{Cross evaluation of energy consumption of ACC, CCC without additional delay and CCC with delay. In observation step, power spectral density is chosen from oracle knowledge~(panels~(a), (b)), periodogram estimator~(panels~(c), (d)) and Welch's method~(panels~(e), (f)), respectively.}
    \label{fig:acc_cnd_ccc_comparison}
\end{figure}

\begin{figure}[h]
\centering
\includegraphics[width=0.55\textwidth]{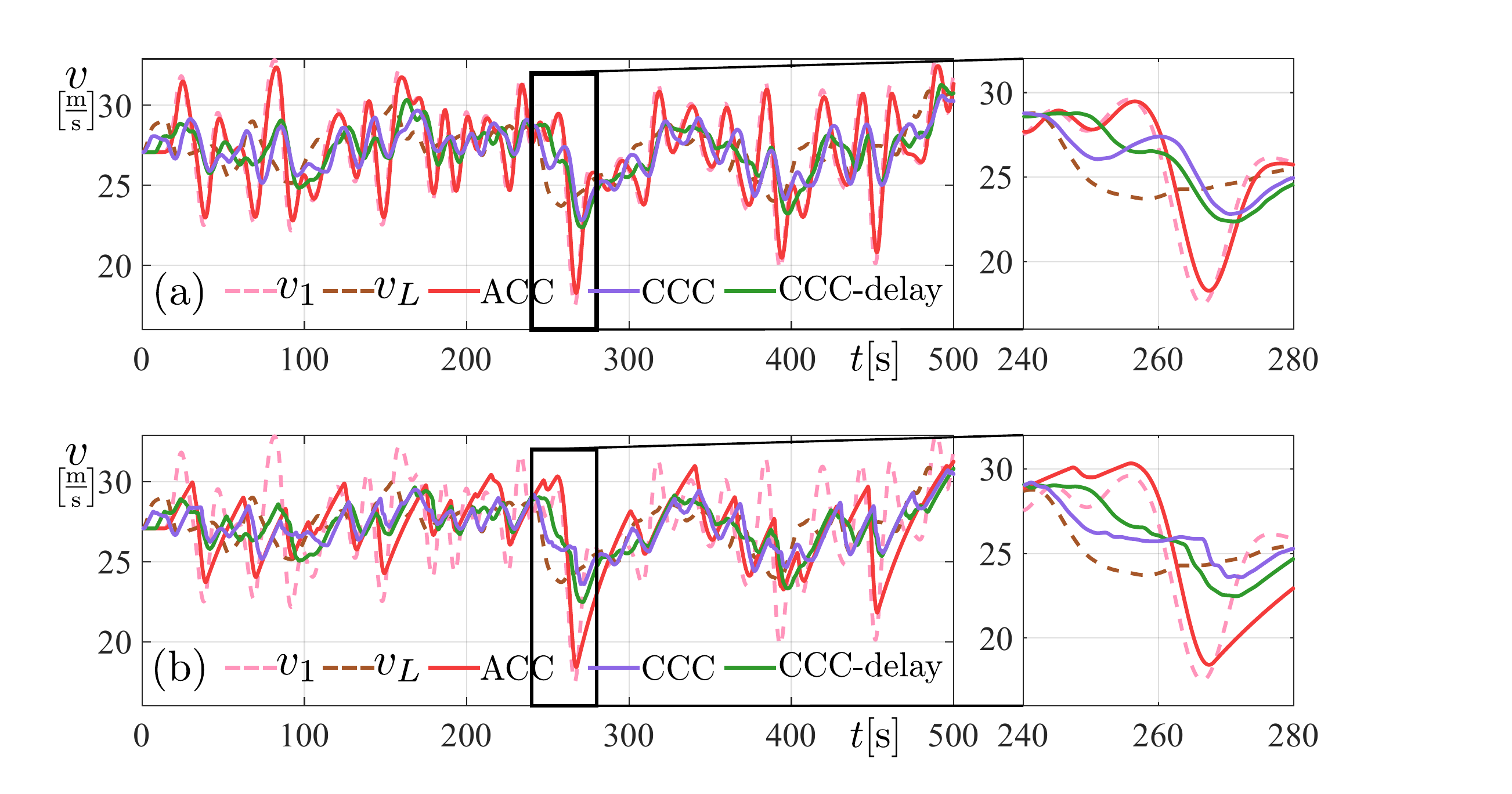}
\caption{Speed profile of the CAT following synthetic trajectory in Fig.~\ref{fig:traj_real_synthetic}(c) using ACC, CCC and CCC with additional delay. We choose vehicle $8$ as leading vehicle~(${L=8}$). Spectral densities are chosen from oracle knowledge. (a) Linear simulation. (b) Nonlinear simulation.}
\label{fig:traj_syn_acc_ccc}
\end{figure}

\begin{table}[h]
\begin{tabular}{cccc}
            & ACC {[}kJ/kg{]} & CCC {[}kJ/kg{]} & CCC-Delay {[}kJ/kg{]} \\ \hline
Oracle      & $3.018$   & $1.331~(-55.89\%)$  & $1.178~(-60.98\%)$ \\
Periodogram & $3.044$   & $1.529~(-49.77\%)$  & $1.378~(-54.72\%)$ \\
Welch       & $3.020$   & $1.379~(-54.35\%)$  & $1.244~(-58.80\%)$                
\end{tabular}
\caption{Average energy consumption $\bar{w}$ in linear simulation.}
\label{tab:avg_energy_lin}
\end{table}

\begin{table}[h]
\begin{tabular}{cccc}
            & ACC {[}kJ/kg{]} & CCC {[}kJ/kg{]} & CCC-Delay {[}kJ/kg{]} \\ \hline
Oracle      & $4.438$  & $3.709~(-16.44\%)$  & $3.635~(-18.11\%)$  \\
Periodogram & $4.453$  & $3.771~(-15.33\%)$  & $3.705~(-16.82\%)$  \\
Welch       & $4.436$  & $3.712~(-16.34\%)$  & $3.653~(-17.66\%)$
\end{tabular}
\caption{Average energy consumption $w$ in nonlinear simulation.}
\label{tab:avg_energy_nonlin}
\end{table}

Fig.~\ref{fig:acc_cnd_ccc_comparison} shows the cross evaluation result for the energy consumption of adaptive cruise control~(ACC), connected cruise control without additional information delay~(CCC), and CCC with delay~(CCC-Delay). In the observation step, power spectral density can be chosen from oracle knowledge~(panel~(a) and (b)), periodogram estimator~(panel~(c) and (d)) and Welch's method~(panel~(e) and (f)), respectively. For both estimators, information from V2X connectivity can bring significant energy reduction compared to ACC, where no V2X connectivity is available. The average energy consumption is compared for the different cases in Tables~\ref{tab:avg_energy_lin} and \ref{tab:avg_energy_nonlin}. In the linear case~(Table~\ref{tab:avg_energy_lin}), connectivity can reduce the energy consumption by half, while in the nonlinear case, at least $15\%$ energy is saved. Note that these savings are achieved by adding a single connected vehicle in the traffic flow. The additional delay leads to additional $5\%$ energy saving in the linear case, and additional $2\%$ saving in the nonlinear case.

To further investigate how connectivity and the additional delay benefits energy consumption, in Fig.~\ref{fig:traj_syn_acc_ccc}, we plot the speed profiles of the CAT executing ACC, CCC, and CCC with additional delay, respectively. We choose synthetic speed trajectories shown in Fig.~\ref{fig:traj_real_synthetic}(c) for preceding vehicles, and vehicle $8$ is the leading vehicle~(${L=8}$). Simulation results of linearized system and nonlinear system are shown in panels~(a) and (b) of Fig.~\ref{fig:traj_syn_acc_ccc}, respectively. The ACC controller closely follows the trajectory of its immediate predecessor. As highlighted by the zoom-ins, there is heavy braking around ${t=260~[\mathrm{s}]}$, which results in heavy braking with the ACC controller. Therefore ACC consumes significant energy. However, with connectivity, the CAT has access to the states of preceding vehicles in the distance. Hence the truck knows that the leading vehicle $L$ brakes around ${t=240~[\mathrm{s}]}$, and using CCC it can brake in advance, creating enough safe margin to avoid heavy braking at ${t=260~[\mathrm{s}]}$. Moreover, since the braking behavior takes around 10 seconds to propagate from vehicle $L$ to $1$, the truck does not need to react immediately to the brake. Instead, we can purposely delay the reaction with a few seconds. Thus, the CCC with additional delay further reduces the speed perturbation.

\subsection{Choice of Spectral Estimator}\label{sec:choice_estimator}

In this section, we show that choosing better spectral estimator can help us get closer to the energy-optimal parameters and reduce the energy consumption. There is a fundamental trade-off between the variance and frequency resolution for spectral estimators~\cite{shumway2017time}. Welch's method has less variance than periodogram, at the cost of lower frequency resolution. We compare these two spectral estimators by simulations using synthetic data described in Section \ref{sec:dataset}.

Figure~\ref{fig:acc_traj_stats} illustrates the performance of the spectral estimators. 
Panel~(a) shows the mean of the sample autocorrelation function, which is in excellent agreement with the oracle \eqref{eq:matern_kernel}.
Panel~(b) compares the power spectral density calculated by the periodogram (orange curve and shading) and Welch's method (green curve and shading) with the oracle power spectral density (purple dashed curve). The latter is obtained as the Fourier transform of \eqref{eq:matern_kernel}. The means of spectral estimators match the oracle very well except at zero frequency. Note that the spectral density at zero frequency does not influence the objective function in \eqref{eqn:optimization_continuous}. The periodogram estimator has higher resolution compared to Welch's method, but the variance is significantly higher. In practice, Welch's method demands more data but, when stationary assumptions hold for a long enough time, this method can bring more precise spectral description.

\begin{figure}
    \centering
    \includegraphics[width=0.44\textwidth]{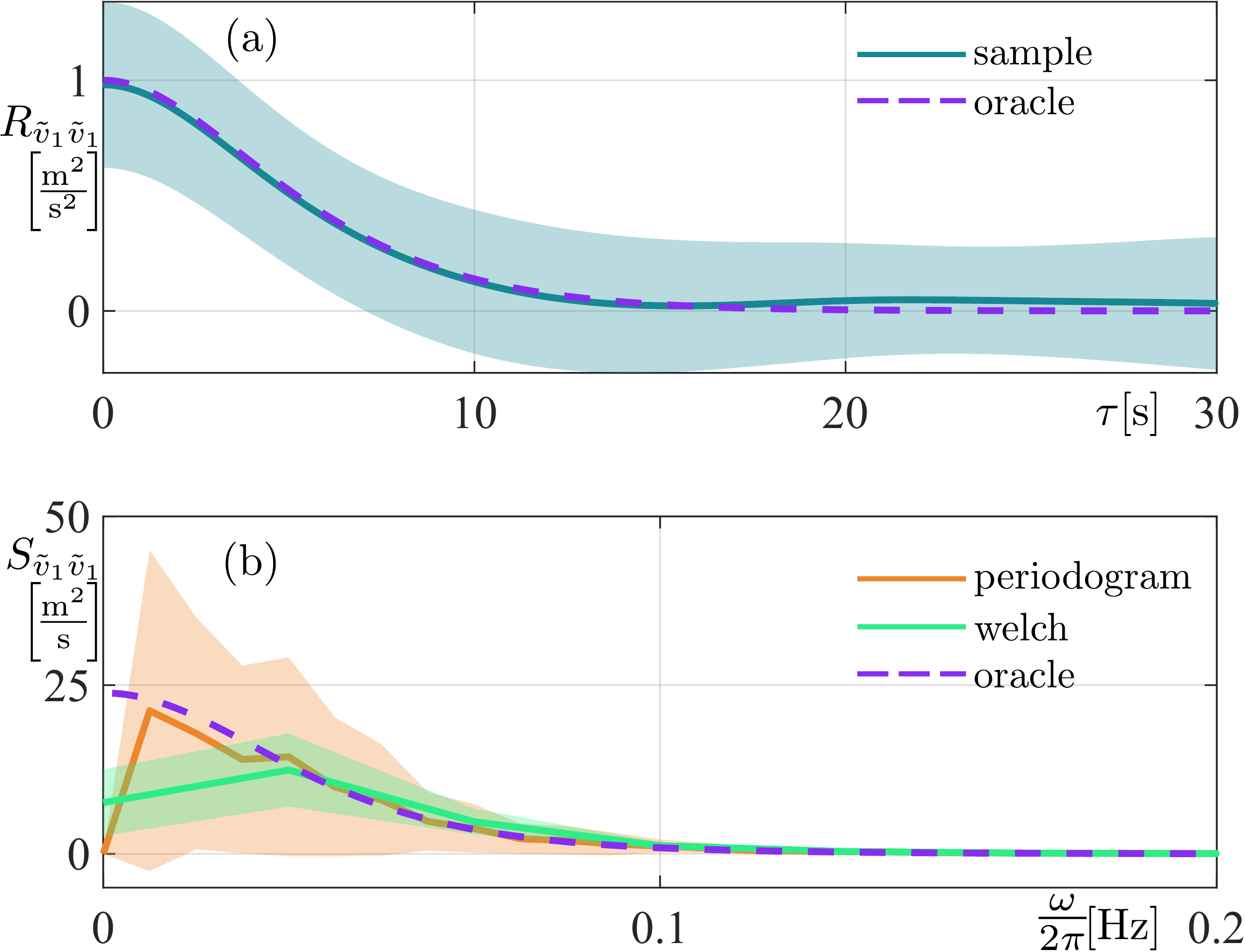}
    \caption{Comparison of the periodogram and Welch's spectral estimators. (a) Sample-based correlation function and the oracle correlation function \eqref{eq:matern_kernel}. (b) Spectral estimations and the oracle power spectral density. Solid curves denote the mean while the standard deviation is indicated by shading.}
    \label{fig:acc_traj_stats}
\end{figure}

Fig.~\ref{fig:synthetic_params} compares the controller parameters chosen from oracle, periodogram and Welch's method, denoted as $\theta^{(\mo)}$, $\theta^{(\mp)}$, and $\theta^{(\mw)}$ respectively, where ${\theta\in \{\beta_1, \beta_L, \sigma_L\}}$. One may observe that Welch's method achieves better concentration around the oracle parameters than the periodogram, and the resulting parameters lie closer to the oracle which is the optimum based on ground truth distribution.

\begin{figure}
    \centering
    \includegraphics[width=0.40\textwidth]{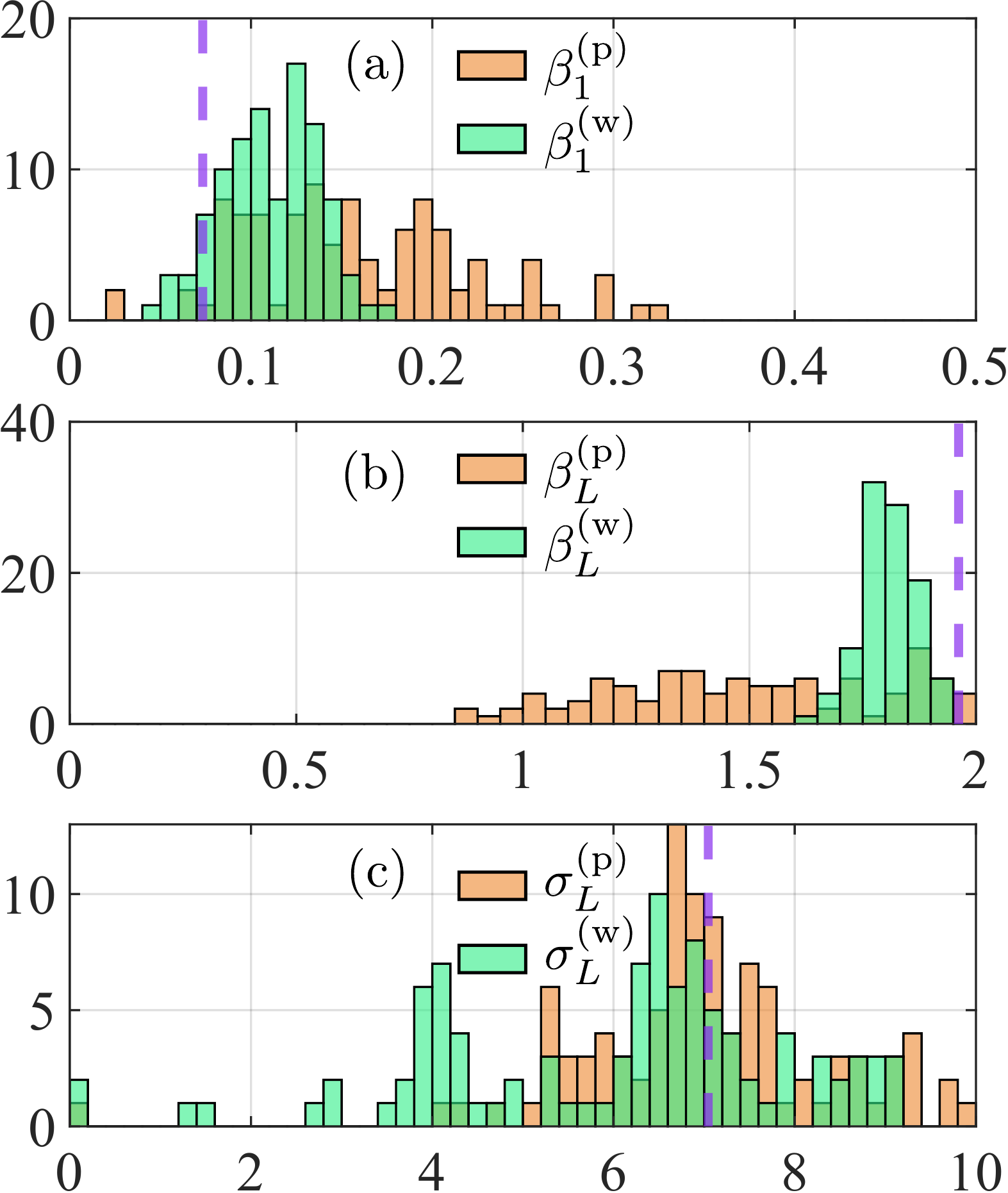}
    \caption{Distribution of controller parameters optimized for energy consumption.
    Orange histogram corresponds to the periodogram estimator while green corresponds to Welch's method. Purple dashed lines correspond to the optimal controller parameters chosen with oracle knowledge of speed distribution.}
    \label{fig:synthetic_params}
\end{figure}

Finally, we compare the energy consumption of the linearized and nonlinear dynamics using parameters chosen from oracle, periodogram and Welch's method. In the linear case, the energy consumption is evaluated using surrogate model~\eqref{eqn:linear_energy_model} which neglects the nonlinear physical model. We denote the corresponding energy consumptions as $\bar{w}^{(\mo)}$, $\bar{w}^{(\mp)}$ and $\bar{w}^{(\mw)}$. In the nonlinear case, the energy consumption is evaluated using~\eqref{eqn:w_definition}, and the corresponding energy consumptions are ${w^{(\mo)}}$, ${w^{(\mp)}}$ and ${w^{(\mw)}}$. To compare the three spectral estimations, we compute the relative energy advantages
\begin{subequations}\label{eq:ws}
\begin{equation}
{\Delta \bar{w}^{(\mathrm{\lozenge\square})} = (\bar{w}^{(\lozenge)} - \bar{w}^{(\square)})/\bar{w}^{(\square)}}~,
\end{equation}
\begin{equation}
{\Delta w^{(\mathrm{\lozenge\square})} = (w^{(\lozenge)} - w^{(\square)})/w^{(\square)}}~,
\end{equation}
\end{subequations}
where ${\lozenge, \square \in \{\mathrm{o, p, w}\}}$. The histograms of 10100 observation-testing pairs are shown in Fig.~\ref{fig:synthetic_pwchoice}. The panels in the left column show linear results, while in the right column we show nonlinear results. In panels (a) and (b), we compare the energy consumptions of the estimators to those of the oracle. The distribution of energy consumption is more concentrated around $0$ for Welch's method, which means in most cases, the parameter given by Welch's method can achieve similar energy consumption as the benchmark oracle parameter. We also compare these two spectral estimators directly in panels (c) and (d). For most of the cases, Welch's method consumes less energy than periodogram as ${\Delta \bar{w}^{(\mathrm{pw})}}$ and ${\Delta w^{(\mathrm{pw})}}$is distributed more towards positive values. On average, the periodogram parameters consume $10.78\%$ more energy than Welch parameters in linear case while $1.42\%$ more energy in nonlinear case.

\begin{figure}
    \centering
    \includegraphics[width=0.48\textwidth]{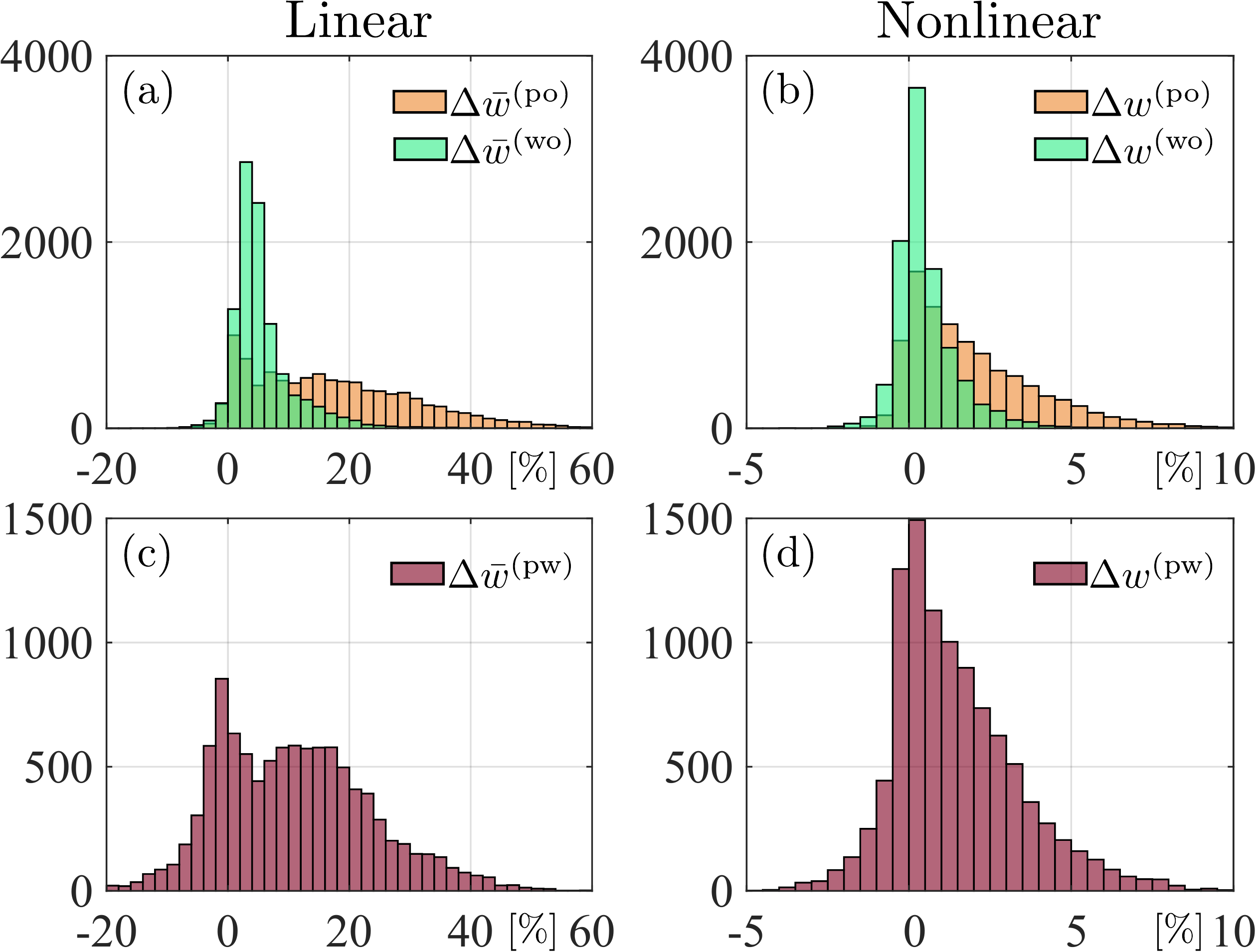}
    \caption{Comparison of energy consumption of the periodogram and Welch's estimator. (a),~(b) Energy consumption compared to the oracle. (c),~(d) Direct comparison of the two spectral estimators. }
    \label{fig:synthetic_pwchoice}
\end{figure}

\subsection{Benefits of Additional Delay}\label{sec:benefit_delay}

Here we quantify the benefits of incorporating the additional delays $\sigma_i$ into the controller~\eqref{eqn:ccc_ad}.
Recall that these delays were introduced based on the following intuition.
Considering lean penetration of connected vehicles, the CAT may connect to vehicles far in the distance. Introducing additional delays enables the CAT to ``wait" until the effect of the distant vehicles' motion propagate closer, and thus, it may achieve a more energy-efficient response.
In this section, we compare the controller with and without additional delay using synthetic data as well as experimental data.

First, we apply the proposed method to find the optimal parameters for the case when ${\sigma_i=0}$. The corresponding energy consumption values using oracle, periodogram and Welch's method are denoted by $\hat{w}^{(\mo)}$, $\hat{w}^{(\mp)}$ and $\hat{w}^{(\mw)}$, respectively. In order to compare this controller with the one with additional time delay, we define the relative energy advantages
\begin{subequations}
\begin{equation}
\Delta \hat{\bar{w}}^{(\lozenge)} = (\hat{\bar{w}}^{(\lozenge)} -  \bar{w}^{(\lozenge)})/\bar{w}^{(\lozenge)}~,
\end{equation}
\begin{equation}
\Delta \hat{w}^{(\lozenge)} = (\hat{w}^{(\lozenge)} -  w^{(\lozenge)})/w^{(\lozenge)}~,
\end{equation}
\end{subequations}
for linear and nonlinear cases, respectively, where ${\lozenge\in \{\mo, \mp, \mw\}}$, cf.~\eqref{eq:ws}. The corresponding histograms are shown in Fig.~\ref{fig:synthetic_delay_compare}. For all methods used, the additional delay brings energy benefits. In the linear case, oracle, periodogram and Welch's method save $11.53\%$, $9.86\%$ and $9.77\%$ energy on average. In the nonlinear case, the average energy benefits are $2.00\%$, $1.76\%$ and $1.59\%$, respectively. Though the nonlinearity in the dynamics impairs the advantage, the difference is still significant.

\begin{figure}
    \centering
    \includegraphics[width=0.44\textwidth]{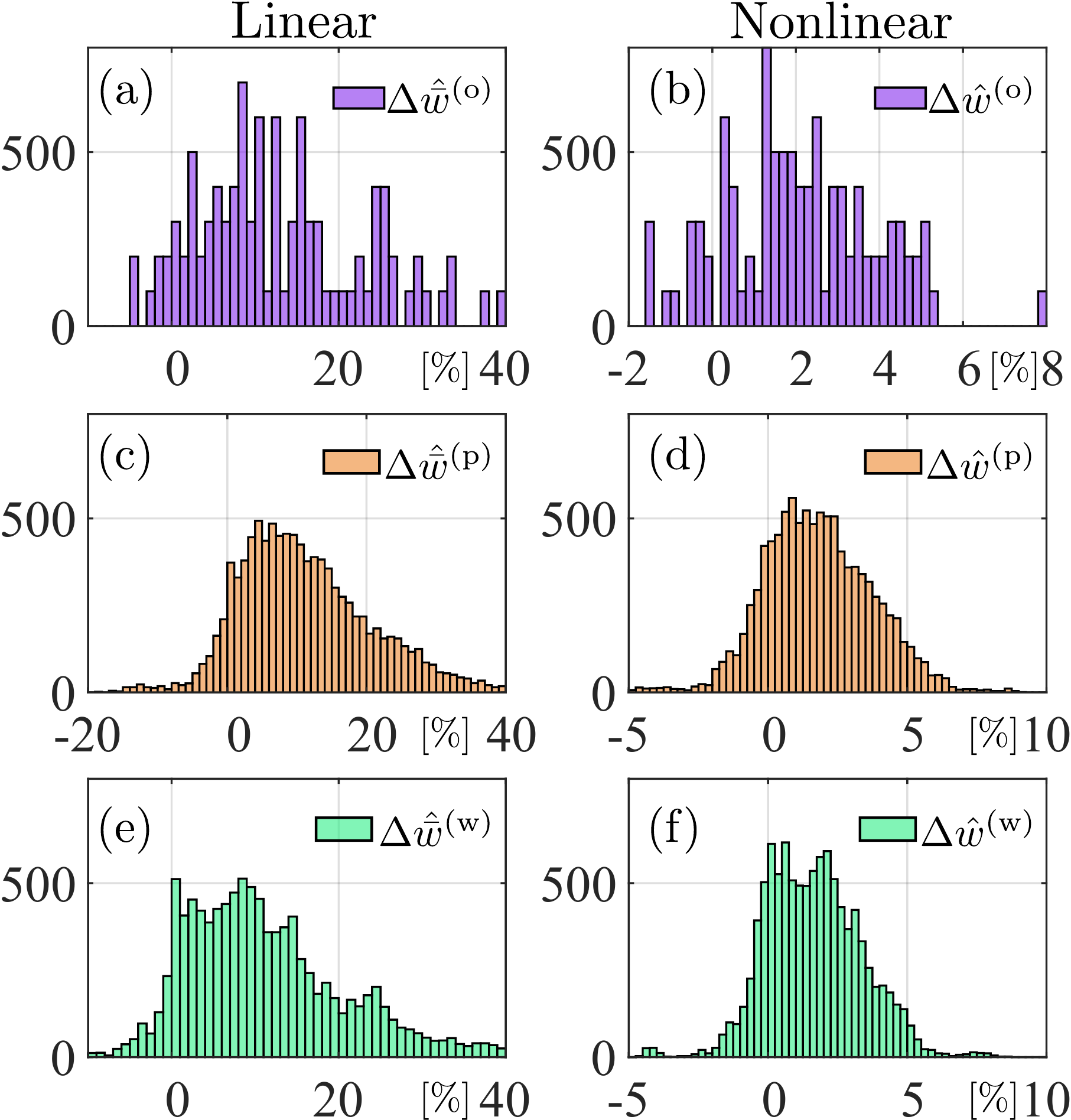}
    \caption{Comparison of controllers with and without additional delay. The histograms are distributed towards positive values in each panel, which indicates that the additional delay leads to energy savings.}
    \label{fig:synthetic_delay_compare}
\end{figure}


For lean penetration of connectivity, the number of vehicles driving between the CAT and the leading CHV may be varying. In previous simulations, we fixed the leading vehicle to ${L=8}$. Now, we investigate optimal controller parameters and the corresponding energy consumption for different leading vehicles ${L=2, \ldots, 8}$ and show that our method is agnostic to the change of leading vehicle.

\begin{figure}[t!]
    \centering
    \includegraphics[width=0.38\textwidth]{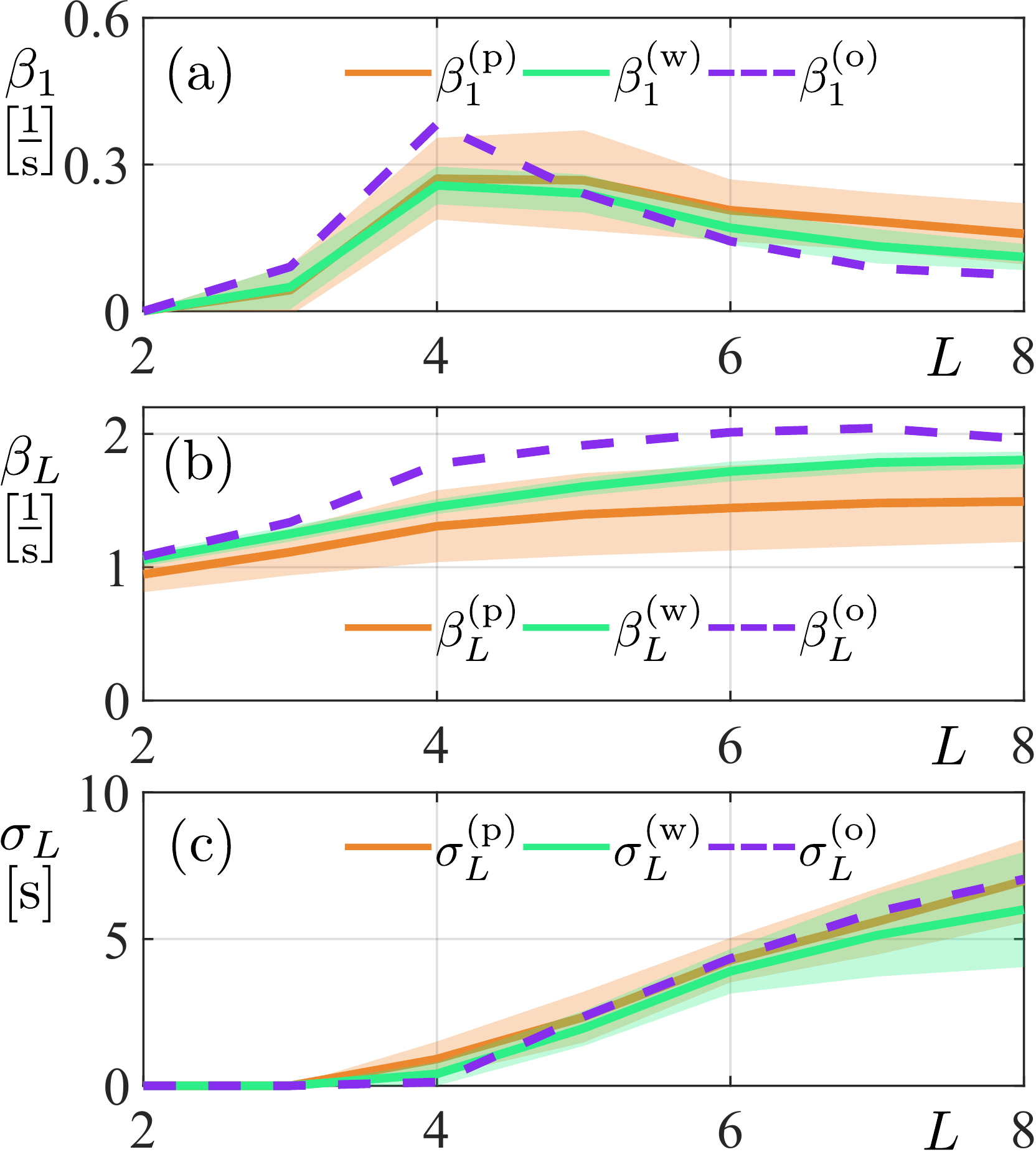}
    \caption{Controller parameters for varying leader (${L=2,\cdots, 8}$ vehicles ahead of the CAT). Each synthetic dataset produces a controller parameter triplet ${(\beta_1, \beta_L, \sigma_L)}$. The means of parameters optimized with periodogram and Welch's method are plotted with solid line. The widths of shaded areas are determined by the standard deviations. Controller parameters from oracle knowledge of spectral density are plotted with dashed lines.}
    \label{fig:varL_param}
\end{figure}

\begin{figure}[b]
    \centering
    \includegraphics[width=0.48\textwidth]{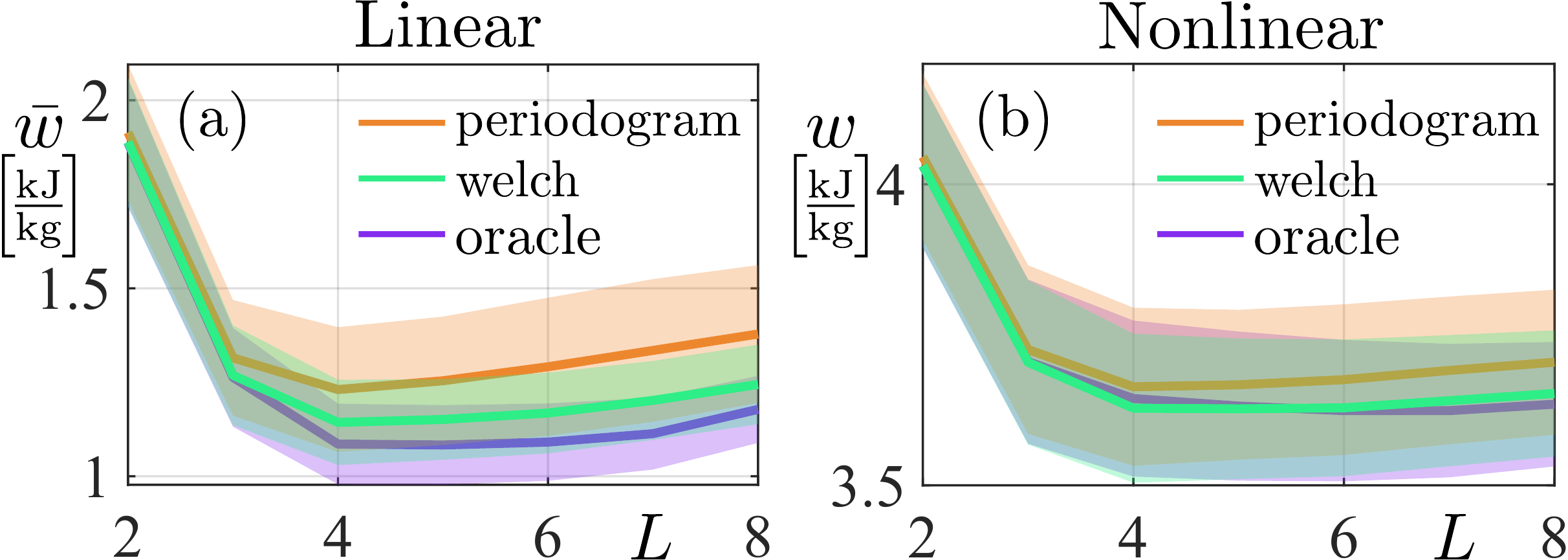}
    \caption{Energy consumption for varying leader (${L=2,\cdots, 8}$) using 10100 observation-evaluation pairs. Solid lines represent the mean values while the width of shaded areas indicate the standard variations.}
    \label{fig:varL_w}
\end{figure}

In Fig.~\ref{fig:varL_param}, we show optimized controller parameters for different leading vehicles. For each of the 101 synthetic datasets, spectral densities are estimated with periodogram and Welch's method resulting in 101 controller parameter sets ${(\beta_1^{(\lozenge)}, \beta_L^{(\lozenge)}, \sigma_L^{(\lozenge)}), \lozenge\in\{\mathrm{p,w}\}}$. The mean of the parameters are plotted with solid line, and the width of shaded area indicates standard deviation. Since the datasets are synthetic, we have access to the oracle knowledge of spectral density. The correspondingly optimized oracle parameters are plotted with dashed line. The mean periodogram and Welch parameters are close to the oracle parameters, and the Welch parameters have smaller deviation for $\beta_1$ and $\beta_L$. Also note that when $L$ is small~(${L=2,3,4}$) the optimal delays ${\sigma_L=0~\mathrm{[s]}}$. This can be explained intuitively: when leading vehicle $L$ is close to the ego vehicle, instantaneous response is preferred without additional waiting time \cite{he2019fuel}.

The energy consumption results with respect to periodogram, Welch and oracle parameters are plotted in Fig.~\ref{fig:varL_w}. The means of energy consumption are plotted with solid line and the standard deviation determines the width of shaded areas. In linear case, shown in panel (a), oracle parameters consume lower average energy than periodogram and Welch parameters, while in nonlinear case depicted in panel (b), Welch parameters have similar and sometimes better average performance as oracle parameters. In both cases, Welch parameters have lower average energy consumption than periodogram parameters. In addition, connecting to vehicles farther in the distance saves more energy than connecting to vehicles nearby due to the string instability of human-driven vehicles ahead. In other words, vehicles in the distance may have lower speed variations, which provides smoother reference trajectories for the controller. 

We make a further case study on the experimental traffic congestion data shown in Fig.~\ref{fig:traj_real_synthetic}(b). We show the optimal energy consumption as a function of the leading vehicle's index $L$ as well as the additional delay $\sigma_L$. For each fixed value of $\sigma_L$, we optimize for $\beta$ and $\beta_L$ using periodogram and Welch's method. The corresponding energy consumptions are plotted in Fig.~\ref{fig:delay_sensitivity}, and the optimal delays $\sigma_L$ chosen by periodogram and Welch's method are marked with crosses.  When the CAT is connected to vehicles nearby, for example $L=2, 3, 4, 5$, the additional delay does not bring extra energy benefits, since the propagation time of the congestion waves between vehicle $L$ and the CAT is short. However, for more distant connections, such as $L = 6, 7, 8$, incorporating the additional delay $\sigma_{\mathrm{L}}$ yields significant  energy savings. This is consistent with results in Fig.~\ref{fig:varL_param}(c). Furthermore, connecting to vehicles farther in the distance leads to more energy benefits than connecting to vehicles nearby, which is consistent with Fig.~\ref{fig:varL_w}.

\begin{figure}[t]
    \centering
    \includegraphics[width=0.48\textwidth]{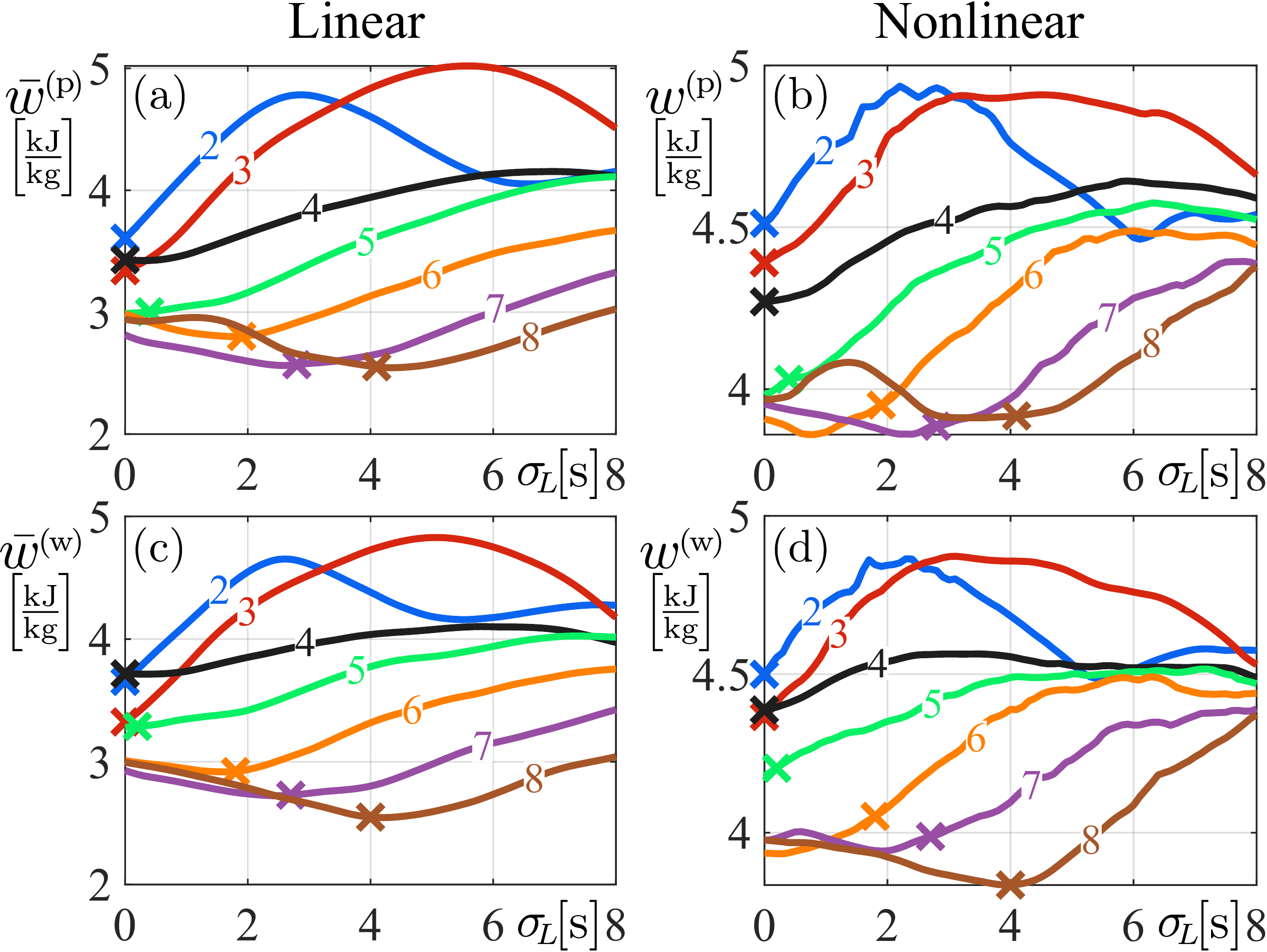}
    \caption{Energy consumption as a function of the additional delay $\sigma_L$ in the controller for different leading vehicles~($L=2,\cdots, 8$), for the case of a traffic congestion dataset. Crosses correspond to the optimal $\sigma_L$ values based on periodogram and Welch's method. Panels~(a)(c) show the energy consumptions for the linear case. }
    \label{fig:delay_sensitivity}
\end{figure}

\section{Conclusion}\label{sec:concl}

In this paper, we designed longitudinal controllers for a connected automated truck traveling in mixed traffic that consists of connected and non-connected vehicles. We leveraged that the truck has access to beyond-line-of-sight information via vehicle-to-vehicle communication, and we introduced an additional delay in the control law when responding to distant connected vehicles.
Human-driven traffic was modeled by stationary stochastic processes and car-following models, where the spectral properties of the stochastic processes were linked to the average energy consumption with a new theorem. The controllers were optimized by minimizing average energy consumption. In the underlying optimization problem, the spectral density of the stochastic process was estimated from data using spectral estimators. We showed that our optimization framework can select designs with significant energy saving. It can also facilitate improvements when utilizing motion information from distant vehicles.
Simulations with large amount of synthetic data showed that energy benefits can be realized even with lean penetration of connected vehicles, regardless of their positions in the traffic. The theory in this paper is mainly based on linear systems under stationary assumptions. It can be readily applied not only to truks but other types of vehicles independent of their propoltion system. Future research will analyze the effects of nonlinearity and transients on the controller performance, and adapt the control parameters online in real traffic.


\ifCLASSOPTIONcaptionsoff
  \newpage
\fi

\bibliographystyle{IEEEtran}
\bibliography{ref.bib}

%




\end{document}